\documentclass[11pt]{article}
\usepackage{amsfonts}%
\usepackage{fullpage}
\usepackage{subcaption} 
\usepackage{booktabs} 
\usepackage{tikz}
\usepackage{amsmath}
\usepackage[noend]{algpseudocode}
\usepackage{algorithm}
\usepackage{enumerate}
\usepackage{mathrsfs}  
\usepackage{url}

\usepackage{amsthm}

\newtheorem{theorem}{Theorem}

\newtheorem{lemma}{Lemma}
\usepackage{moresize}

\usepackage{array}
\algblockdefx{FORALLP}{ENDFAP}[1]%
  {\textbf{parallel for}#1 \textbf{do}}%
  {\textbf{end parallel for}}

\usetikzlibrary{positioning}
\usetikzlibrary{patterns}

\algnewcommand{\algorithmicand}{\textbf{ and }}
\algnewcommand{\algorithmicor}{\textbf{ or }}
\algnewcommand{\OR}{\algorithmicor}
\algnewcommand{\AND}{\algorithmicand}
\algnewcommand{\var}{\texttt}

\algblockdefx{FORALLP}{ENDFAP}[1]%
  {\textbf{parallel for}#1 \textbf{do}}%
  {\textbf{end parallel for}}

\newcolumntype{L}[1]{>{\raggedright\let\newline\\\arraybackslash\hspace{0pt}}m{#1}}
\newcolumntype{C}[1]{>{\centering\let\newline\\\arraybackslash\hspace{0pt}}m{#1}}
\newcolumntype{R}[1]{>{\raggedleft\let\newline\\\arraybackslash\hspace{0pt}}m{#1}}

\algnewcommand\algorithmicinput{\textbf{Input:}}
\algnewcommand\INPUT{\item[\algorithmicinput]}
\algnewcommand\algorithmicoutput{\textbf{Output:}}
\algnewcommand\OUTPUT{\item[\algorithmicoutput]}

\algdef{SE}[DOWHILE]{Do}{doWhile}{\algorithmicdo}[1]{\algorithmicwhile\ #1}%
\algnewcommand\algorithmicforeach{\textbf{for each}}
\algdef{S}[FOR]{ForEach}[1]{\algorithmicforeach\ #1\ \algorithmicdo}

\AtBeginDocument{%
  }

\title{Speeding-up Graph Algorithms via Clique Partitioning}

\usepackage{authblk}

\author[1]{Akshar Chavan}
\author[2]{Sanaz Rabinia}
\author[2]{Daniel Grosu}
\author[1]{Marco Brocanelli}

\affil[1]{Dept. of Electrical and Computer Engineering, The Ohio State University, Columbus, OH, USA\\
\texttt{chavan.43@osu.edu, brocanelli.1@osu.edu}}

\affil[2]{Dept. of Computer Science, Wayne State University, Detroit, MI, USA\\
\texttt{srabin@wayne.edu, dgrosu@wayne.edu}}

\date{}

\begin{document}

\maketitle
\begin{abstract}
Reducing the running time of graph algorithms is vital for tackling real-world problems such as shortest paths and matching in large-scale graphs, where path information plays a crucial role. 
To address this critical challenge, this paper introduces a graph restructuring algorithm that identifies bipartite cliques and replaces them with tripartite graphs. This restructuring leads to fewer edges while preserving complete graph path information, enabling the direct application of algorithms like matching and all-pairs shortest paths to achieve significant runtime reductions, especially for large, dense graphs.
The running time of the proposed algorithm for a graph $G(V,E)$, with $|V| = n$ and $|E| = m$
is~$O(mn^\delta)$, which is better than $O(mn^\delta \log^2 n)$, the running time of the best existing algorithm for speeding-up other graph algorithms (the Feder-Motwani (\textsf{FM}) algorithm),  where $0 \leq \delta \leq 1$. Both the \textsf{FM} algorithm and the proposed algorithm are originally formulated for bipartite graphs, but can also be applied to general directed or undirected graphs.
Our extensive experimental analysis demonstrates that the proposed algorithm achieves up to 21.26\% higher reduction in the number of edges and runs up to 105.18× faster than the \textsf{FM} algorithm. On large synthetic graphs with up to 1.05 billion edges, it attains a reduction in the number of edges of up to 74.36\%. On real-world graphs, it achieves a reduction in the number of edges by up to 46.8\%. Furthermore, when used as a preprocessing step, our approach yields up to a 2.07× speedup for the matching algorithms on large synthetic graphs, and up to a 1.74× speedup for the All-Pairs Shortest Path algorithms on real-world graphs, when compared to using the given graph as input. \looseness -1

\noindent\textbf{keywords:} graph algorithms, speeding-up graph algorithms, clique partitioning
\end{abstract}





\section{Introduction}
\label{sec:intro}

Graph algorithms are essential for analyzing complex real-world networks across diverse domains like social interactions, biological systems, and communication infrastructures. However, as these networks grow in size and complexity, the time required to process them using graph algorithms becomes excessively high. This paper tackles the critical issue of high computation time in large-size graphs by introducing a graph restructuring approach that speeds up the execution of graph algorithms such as matching and all pairs shortest paths. 
This approach involves transforming the graph into a more computationally suitable structure, preserving its path information, where path information refers to the existence of paths between vertices. Specifically, by restructuring the graph to reduce the number of edges, we aim to speeding-up graph algorithms. A key requirement for this is the preservation of complete path information during the graph restructuring process. Many crucial graph algorithms, such as matching and all-pairs shortest paths (\textsf{APSP}), rely on this global connectivity information to produce correct results. \looseness -1

However, to effectively speed up the downstream algorithms, the restructuring algorithm itself must be efficient and the resulting structure readily usable. Specifically, it must meet two critical conditions: 1) the graph restructuring algorithm must have a low running time to ensure that its use as a preprocessing step leads to an overall reduction in execution time compared to running the algorithms on the original graph; and 2) the restructured graph should be directly usable as input to these algorithms or require minimal modifications. These conditions highlight the inherent challenges in designing effective graph restructuring algorithms for speeding-up graph algorithms. In this paper, we address these challenges and propose a new  Clique Partitioning based Graph Restructuring (\textsf{CPGR}) algorithm that restructures the graph to reduce the number of edges while preserving the path information. We theoretically and experimentally show that our algorithm achieves faster running time and greater  reduction in the number of edges compared to the algorithm proposed by Feder and Motwani~\cite{federMotwani}, i.e., the best existing  algorithm used to speeding-up the execution of graph algorithms. \looseness -1

\subsection{Related Work}
\label{sec:related_works}

\begin{table}[b!]

\centering
\caption{Overview of Graph Restructuring Methods} 
\small
\begin{tabular}{p{8em} p{8em} p{8em} p{8em} p{8em}}
\toprule
\textbf{Technique} & \textbf{Goal} &\textbf{ Mechanism}  & \textbf{Applications} & \textbf{Limitations\hspace*{0.1cm}for Speeding-up\hspace*{0.1cm}Graph $\qquad \qquad $ Algorithms} \\
\toprule
Sparsification \cite{ijcai2024p891, John_Safro_2016, 08074489X} & Reduce \# of Edges & Edge\hspace*{0.1cm}sub-sampling or filtering & Graph $\qquad \qquad $ Algorithms, $\qquad \qquad $Visualization & Accuracy loss, potential disconnection  \\ \midrule
Coarsening \cite{chen2022graph, chen2023gromov} & Reduce \# of Vertices & Vertex grouping / aggregation & Multilevel\hspace*{0.1cm}Methods, GNN Scaling  & Loss of fine-grained detail, interpolation error  \\ \midrule
Hierarchical Decomposition \cite{ besta_survey, bulucc2016recent, EPSTEIN2008612, li2017hierarchical}& Reduce Problem Complexity & Subdivide problem using graph structure &  Optimization, Modeling  & Applicability depends on problem structure \\ \midrule
Lossless Compression \cite{DAG_compression, besta_survey, 1scalable, chakrabarti2002locally, compressing_bisection, lossless_contraction} & Reduce Storage Space & Efficient encoding of graph data& Storage, Querying Compressed Data  & Query overhead, less effective on irregular graphs \\ \midrule
Clique Partitioning \cite{federMotwani}, \textbf{\textsf{CPGR}} [this paper] & Speeding-up Graph\hspace*{0.1cm}Algorithms & Extracting\hspace*{0.1cm}bipartite cliques & Graph\hspace*{0.1cm}Algorithms, $\qquad \qquad $ Storage & \multicolumn{1}{c}{--}\\ 
\bottomrule
\end{tabular}

\label{table:current_approaches}
\end{table}

Graph-structured data are common across domains such as social networks, biology, and the web, often reaching billions of edges. To address the computational challenges posed by such scale, graph restructuring techniques have emerged as a key approach for improving the efficiency of graph algorithms through structural transformation.  
These restructuring methods aim to produce graph representations that are enabling fast computation. These approaches include sparsification, graph coarsening, hierarchical decomposition, and lossless compression, each offering trade-offs between speed and accuracy as shown in Table~\ref{table:current_approaches}. \looseness -1

Graph sparsification reduces the edge count while approximately preserving properties like cuts, distances, or spectral characteristics, enabling faster execution of algorithms such as PageRank, at the cost of precision \cite{ijcai2024p891, John_Safro_2016,08074489X}.
Graph coarsening merges nodes and edges to create a smaller approximation of the original graph, commonly used in multilevel algorithms for partitioning and clustering \cite{chen2022graph, chen2023gromov}. While this improves scalability, it can obscure fine-grained structural details.
Hierarchical decomposition organizes graphs into nested clusters or layers, facilitating divide-and-conquer approaches in tasks like community detection and multiscale visualization \cite{ besta_survey,bulucc2016recent, EPSTEIN2008612, li2017hierarchical}. The main drawback lies in the overhead of maintaining and traversing hierarchical structures.
Lossless compression techniques aim to reduce the graph size while preserving exact recoverability, typically using encoding, redundancy elimination, or succinct data structures \cite{DAG_compression, besta_survey, 1scalable, chakrabarti2002locally, compressing_bisection, lossless_contraction}. Although accurate, such methods can add significant overhead. \looseness -1

Despite their advantages, many of these techniques either compromise path-preservation or introduce substantial preprocessing overhead, limiting their applicability to algorithms that require full connectivity information such as matching and APSP.
A more targeted approach is path-preserving graph restructuring, proposed by  Feder and Motwani \cite{federMotwani}, that partitions the graph into bipartite cliques to speeding-up downstream graph algorithms such as bipartite matching, APSP, and vertex connectivity. Their approach guarantees path preservation but suffers from sequential dependencies during clique extraction, which may limit the performance in terms of execution time and reduction in the number of edges. \looseness -1

Building on this foundation, our work proposes a novel Clique Partitioning-based Graph Restructuring (\textsf{CPGR}) algorithm. Unlike Feder and Motwani’s iterative singleton clique extraction, \textsf{CPGR} introduces multiple and larger clique extraction, enabling more efficient restructuring while preserving exact path information. This makes \textsf{CPGR} a promising candidate for preprocessing large-scale graphs to accelerate a wide range of graph algorithms. \looseness -1

\section{Background and Motivation}
\label{sec:b_and_m}
\subsection{Background}
\label{sec:background}

\begin{figure*}[t]
\centering
    \includegraphics[width=0.65\linewidth]{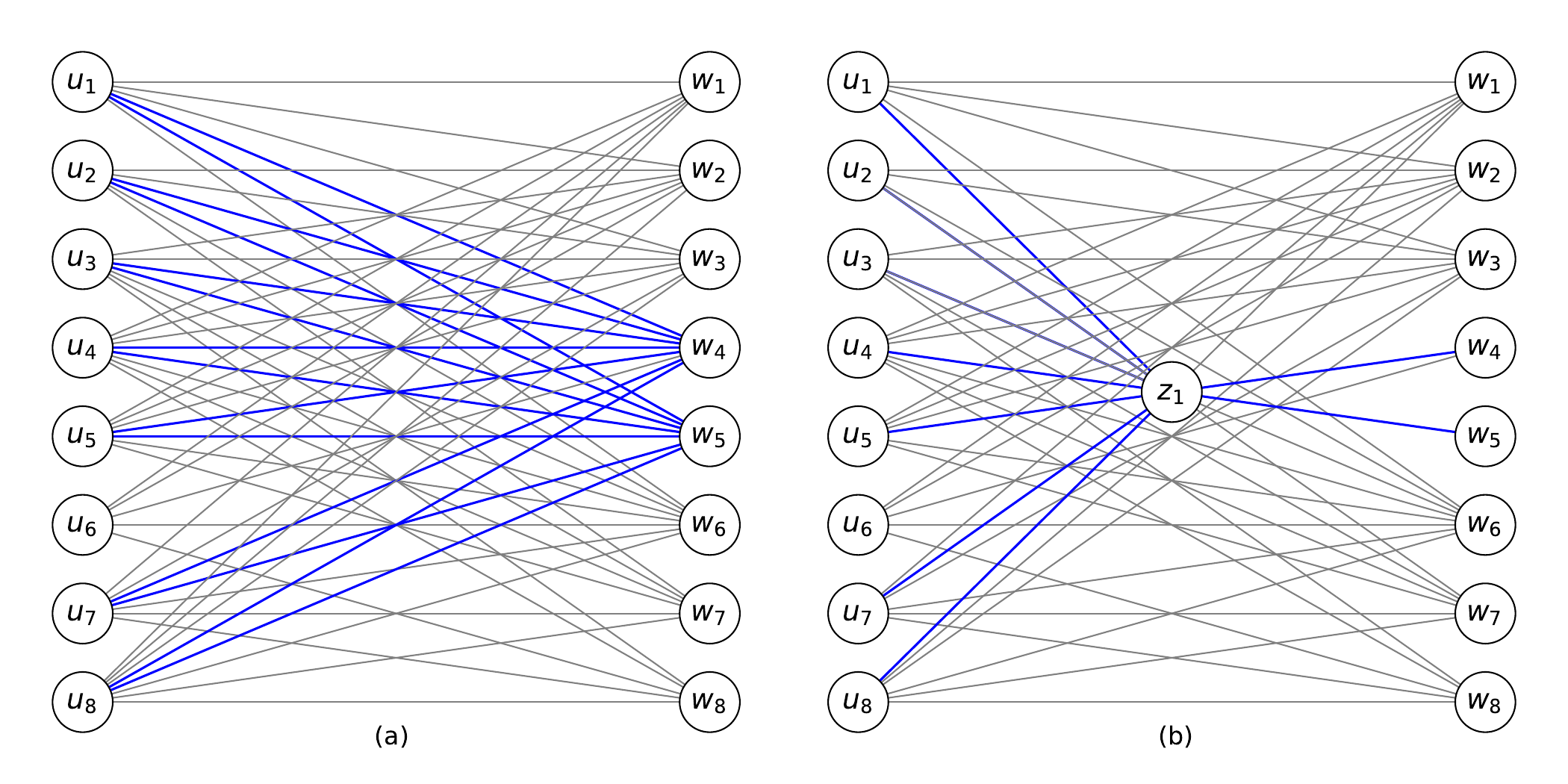}  
      \caption{(a) Given bipartite graph $G(U,W,E)$ and (b) the tripartite graph that replaces the $\delta$-clique with left partition $\{u_1,u_2,u_3, u_4, u_5, u_7, u_8\}$ and right partition $\{w_4, w_5\}$ in the restructured graph $G^{*}(U,W,Z, E^{*})$.}
  \label{fig:given_graph}
\end{figure*}

Feder and Motwani~\cite{federMotwani} proposed a graph restructuring algorithm—referred to as \textsf{FM} in this paper—designed for bipartite graphs~$G(U, W, E)$, where $|U| = |W| = n$ and $|E| = m$. The algorithm identifies complete bipartite subgraphs called $\delta$-\emph{cliques} with left and right partitions of size $\lceil n^{1 - \delta} \rceil$ and $k(n, m, \delta) = \big\lfloor \frac{\delta \log n}{\log(2n^2/m)} \big\rfloor$, respectively, for some constant $0 \leq \delta \leq 1$.
Each $\delta$-clique is then replaced by a more compact tripartite graph that introduces a new vertex~$z_q$, connecting it to all vertices in the clique’s left and right partitions. This transformation reduces the number of edges from $|U_q| \times |W_q|$ to $|U_q| + |W_q|$ per clique, achieving significant reduction in the number of edges while preserving the original graph’s connectivity between $U$ and $W$. For instance, Figures~\ref{fig:given_graph}a and~\ref{fig:given_graph}b show an example of a $\delta$-clique with left partition $\{u_1,u_2,u_3, u_4, u_5, u_7, u_8\}$ and right partition $\{w_4, w_5\}$ in a bipartite graph~$G$, and the corresponding tripartite graph that replaces it in the restructured graph~$G^*$, respectively. The number of edges in the $\delta$-clique is $14 = 7 \times 2$ while the number of edges in the corresponding tripartite graph is $9=7+2$, thus the number of edges in the restructured graph is reduced by~5. Thus, the restructured graph $G^* = (U, W, Z, E^*)$ consists of three disjoint vertex sets, where $Z$ contains one new vertex per extracted clique, and edges in $E^*$ connect vertices across these partitions. Although designed for bipartite graphs, the approach can be extended to general graphs as well.
Further details on the \textsf{FM} algorithm, including the clique extraction process and the use of neighborhood trees, are provided in Appendix~\ref{sec:appendix:FM}. \looseness -1

\subsection{Motivation}
\label{sec:motivation}

We motivate our approach by highlighting three critical aspects related to extracting and replacing $\delta$-cliques that influence both the running time and the reduction in the number of edges. These aspects arise in existing approaches and suggest opportunities for improving how dense substructures are identified and processed.

\vspace*{0.07cm}
\noindent\textbf{Vertex selection.}
A key challenge in identifying $\delta$-cliques for reducing the number of edges lies in selecting vertices for the right partition~$W_q$ that share a large set of common neighbors in~$U$. In approaches such as the \textsf{FM} algorithm, vertices for the right partition are selected using a combinatorial function that accounts for vertex degrees across different levels of a neighborhood tree. At each level of the neighborhood tree, vertex degrees within the corresponding partitions are aggregated to determine which branch to explore. As a result, vertex selection is driven by combined degree contributions across partitions rather than by the individual potential of vertices to form dense bipartite subgraphs with~$U$, which may cause vertices with higher degree to be overlooked when forming dense $\delta$-cliques.

\begin{figure*}[t]
\centering
    \includegraphics[width=\linewidth]{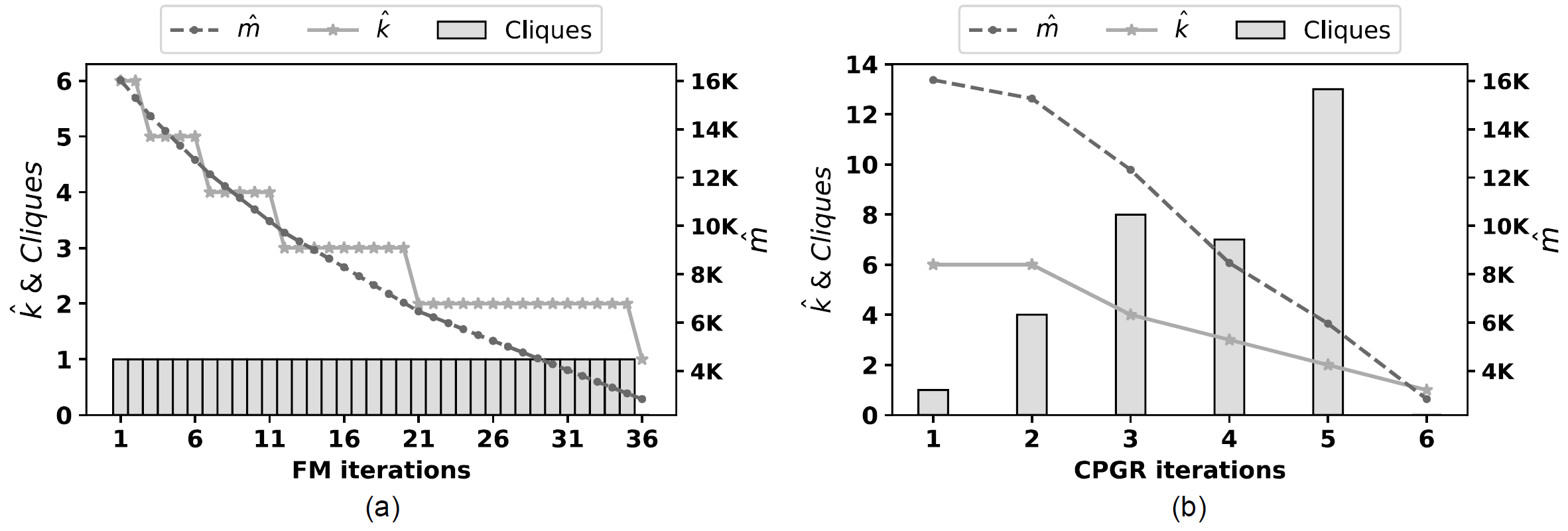}  
    \label{fig:k_hatAndM_hat(fm)}
  \caption{Progression of $\hat{m}$, $\hat{k}$, and number of cliques extracted for a graph with 128 vertices in each bi-partition, density 0.98, and $\delta = 1$ by (a) \textsf{FM}  and  (b) \textsf{CPGR}.}
  \label{fig:motivation}
\end{figure*}

\vspace*{0.07cm}
\noindent\textbf{Extracting $\delta$-cliques.}
The process of forming $\delta$-cliques relies on selecting $\hat{k} = k(n,\hat{m},\delta)$ vertices from the right partition~$W$, where $\hat{m}$ is the number of remaining edges in the graph and $n$ is the number of vertices in each partition. As the algorithm progresses and more cliques are extracted, the number of edges~$\hat{m}$ decreases, which in turn affects the value of~$\hat{k}$. 
Due to the stepwise nature of the function $k(n,\hat{m},\delta)$, the value of~$\hat{k}$ often remains constant across multiple iterations. As illustrated in Figure~\ref{fig:motivation}a, this results in a staircase-like progression in which several iterations operate under the same~$\hat{k}$ while removing relatively few edges per iteration. This behavior suggests that the iterative structure may not fully exploit the potential of each~$\hat{k}$ value, particularly during phases in which $\hat{k}$ remains unchanged over multiple iterations.

\vspace*{0.07cm}
\noindent\textbf{Update mechanism.}
After selecting vertices from the right partition, the manner and timing of degree updates play a crucial role in determining which edges remain available for future $\delta$-clique discovery. In the existing approach (\textsf{FM} algorithm), degree updates are applied incrementally as candidate vertices are selected, rather than after a $\delta$-clique has been fully determined. As a result, degree reductions may be applied to vertices based on partial clique construction, rather than solely on edges that are ultimately removed by an extracted clique. 
Such early updates can reduce the availability of edges that could participate in subsequent $\delta$-cliques, particularly in sparse or moderately dense graphs. This observation highlights the importance of update mechanisms that accurately reflect only the edges eliminated by an extracted clique.
\looseness=-1

\subsection{Our contributions}

This paper introduces a novel algorithm, the Clique Partitioning-based Graph Restructuring (\textsf{CPGR}) algorithm. Our main contributions are summarized as follows:
\begin{enumerate}

    \item We present \textsf{CPGR}, a graph restructuring algorithm that removes multiple $\delta$-cliques within a single iteration, achieving a running time of $O(mn^{\delta})$. This improves upon the $O(mn^{\delta}\log^2 n)$ running time of the \textsf{FM} algorithm.

    \item \textsf{CPGR} allows vertices to participate in multiple clique extractions during the restructuring process. For a fixed $\hat{k}$, the algorithm is designed to extract as many valid $\delta$ cliques as possible in each iteration, which increases the total number of extracted cliques, particularly in large and high density graphs, and leads to a larger number of edges removed per iteration, as illustrated in Figure~\ref{fig:motivation}b.

    \item \textsf{CPGR} achieves an average reduction in the number of edges that is \emph{at least} as good as the \textsf{FM} algorithm.

     \item Experimental analysis show that \textsf{CPGR} achieves up to~21\% higher reduction in the number of edges and runs up to 105.18× faster than the \textsf{FM} algorithm. \looseness -1

    \item Experimentally, we show the practical utility of \textsf{CPGR} as a preprocessing step, achieving a speedup of up to 2.07× for matching and 1.74× for \textsf{APSP} on different graph datasets.
\end{enumerate}

\section{The Proposed Algorithm}
\label{sec:proposed_algorithms}
\subsection{Clique Partitioning-based Graph Restructuring \textsf{(CPGR)} Algorithm}

This section details the design of our Clique Partitioning-based Graph Restructuring (\textsf{CPGR}) algorithm, presented in Algorithm~\ref{alg:CPGR}.
For a given input graph~$G$, \textsf{CPGR} obtains a restructured graph~$G^{*}$ of~$G$ by iteratively finding bipartite cliques (i.e., complete bipartite subgraphs) in~$G(U, W, E)$ and replacing them with corresponding tripartite graphs until finding new bipartite cliques does not contribute to the reduction in the number of edges of~$G$. The size of the right partition of the cliques is determined by~$\hat{k} = k(n,\hat{m}, \delta)$, which guarantees a $\delta$-clique, and the selection of vertices~$w_j \in W$ depends on the degree of vertices~$d_{w_j} = |N(w_j)|$, where $N(w_j)$ denotes the set of neighbors of vertex~$w_j$. 
The partition of~$G$ into bipartite cliques guarantees that each edge in~$G$ is in exactly one clique, i.e., 
$E(C_i) \cap E(C_j) = \emptyset$,  $\forall i, j$ and $ i \neq j$, where $E(C_i)$ is the set of edges of $\delta$-clique~$C_i$. 
The input of \textsf{CPGR} consists of the adjacency matrix~$A$ of~$G$, and a constant~$\delta$. \looseness=-1

\vspace*{0.1cm}
\noindent {\bf Initialization} (Lines~1-5). For the given bipartite graph $G(U,W, E)$, \textsf{CPGR} initializes the index~$q$ of the bipartite cliques extracted from~$G$, the number of vertices, $n$ where $n = |U| = |W|$, the number of edges~$\hat{m}$, and the degree of the vertices ${d_{w_j}}$ $\forall j=1, \ldots,  |W|$,  initialized to~$[0]_{n}$, the zero vector fo size $n$. \textsf{CPGR} also initializes the set of extracted bipartite cliques~$\mathcal{C}$, to the empty set.

\label{subsec:CPGR}
\begin{algorithm}[!t]
\caption{{\small  \textsf{CPGR}: Clique Partitioning-based Graph Restructuring Algorithm}}
{
\begin{algorithmic}[1]
\INPUT {$A$: Adjacency matrix;   }
\Statex	\hspace*{0.55cm}{$\delta$: a constant such that $0\leq \delta \leq 1$. }
\State{$q \leftarrow 0 $}
\State{$n \leftarrow|W|$}
\State{$\hat{m} \leftarrow|E|$}
\State{${d_w} \gets [0]_n$}
\State{$\mathcal{C} \gets \emptyset$}

\For{$i = 1,\ldots, n$}
    \For{$j = 1,\ldots, n$}
        \State{ ${d_{w_j}} \gets  {d_{w_j}} +  E_{{u_i}, w_j}$}
    \EndFor
\EndFor
\State{$\hat{k} \gets \Big\lfloor{\frac{\delta \log n}{\log(2n^2/\hat{m})} }\Big\rfloor$}
\While{$\hat{k} > 1$ } 
    \State{($\mathscr{C}, \hat{A} , \hat{d_w} ) \gets \textsf{CSA(}q, \hat{k}, A, n, d_w\textsf{)}$ }
    \State{$\mathcal{C} \gets \mathcal{C} \cup  \mathscr{C} $}
    \State{$d_w \gets \hat{d_w}$}
    \State{$ A \leftarrow \hat{A}$}
    \State{$\hat{m} \gets \sum_{j=1}^{|W|} d_{w_j}$}
    \State{$ q \leftarrow q+ |\mathscr{C}|$}
    \State{$\hat{k} \gets \Big\lfloor{\frac{\delta \log n}{\log(2n^2/\hat{m})} }\Big\rfloor$}
\EndWhile
\State{$ Z \gets  \emptyset$; $ E^* \gets \emptyset$; $\mathcal{E}_{\mathcal{C}} \gets \emptyset$ }
\For{$i \in \mathcal{C} =    \{\mathscr{C}_1(U_1, W_1), \ldots, \mathscr{C}_{q}(U_{q}, W_{q})\}$}
    \State{$ Z \gets  Z \cup \{z_i\}$}
    \State{$ E^* \gets   E^* \cup \{(u, z_i), (z_i, w) \} \quad \forall u \in U_i, \quad \forall w \in W_i $}
    \State{$\mathcal{E}_{\mathcal{C}} \gets   \mathcal{E}_{\mathcal{C}} \cup E_{i}$}
\EndFor
\State{$E^*   \gets  E^* \cup \{(u, w) | (u, w) \in E \setminus  \mathcal{E}_{\mathcal{C}}\}$}
\State{\textbf{Output:} $G^*(U,Z, W, E^*)$, the restructured graph of~$G$.}
\end{algorithmic}
\label{alg:CPGR}
}
\end{algorithm}

\vspace*{0.1cm}
\noindent \textbf{Partition size} (Lines 6-9). 
\textsf{CPGR} computes the degree of each vertex $w_j \in W$, $d_{w_j} = |N(w_j)|$, where $E_{{u_i}, w_j}$ is~1, if there is an edge between~${u_i}$ and~$w_j$, and 0, otherwise (Lines 6-8). It determines the size of the right partition of the bipartite clique, $\hat{k}$ (Line 9), which guarantees the existence of a $\delta$-clique with $k(n,\hat{m}, \delta)$~\cite{federMotwani}. It is important to say that $\hat{m}=|E|=\sum_{j=1}^{|W|} d_{w_j}$. \looseness=-1

\vspace*{0.1cm}
\noindent \textbf{Clique extraction} (Lines 10-17). \textsf{CPGR} proceeds with extracting $\delta$-cliques until extracting new cliques does not contribute to the reduction in the number of edges. This happens when $\hat{k} = 1$, which results in obtaining trivial bipartite cliques. Thus, the while loop in Lines~10 to~17 is executed until trivial cliques are produced.

\emph{Clique Stripping Algorithm} (\textsf{CSA}), presented in Algorithm~\ref{alg:CSA}, finds $\delta$-cliques in the bipartite graph~$G$. \textsf{CSA} takes the updated index~$q$ of the bipartite clique, the size~$\hat{k}$  of the right partition of the bipartite clique, the adjacency matrix~$A$, the size~$n$ of the left partition of~$G$, and the degree~$d_{w_j} = |N(w_j)|, \ \forall w_j \in W$ as input.
\textsf{CSA}, which is described in Subsection~\ref{subsec:S-CPGR}, returns the set~$\mathscr{C}$  of bipartite cliques that it determined in the current execution, the updated adjacency matrix~$\hat{A}$, and the updated degree of vertices~$\hat{d_w}$. 
n Line~12, \textsf{CPGR} adds the new bipartite cliques to the set of all bipartite cliques~$\mathcal{C}$, updates the degree of vertices~$d_w$ (Line~13), and adjacency matrix~$A$ (Line~14). In Line 15, \textsf{CPGR} updates the number of edges, $\hat{m}$ in the given graph $G(U,W,E)$ after extracting the $\delta$-cliques, and the clique index~$q$ (Line~16). Finally, in Line~17, after the $\delta$-cliques are removed, \textsf{CPGR} updates~$\hat{k}$ for the next iterations. \looseness=-1

\vspace*{0.1cm}
\noindent \textbf{Graph restructuring} (Lines 18-23).
When the while loop (Lines 10-17) terminates, \textsf{CPGR} restructures the graph by adding a new vertex set~$Z$ to the bipartite graph.
Each vertex in the left and right partitions of clique $\mathscr{C}_q$ is connected to a new vertex $z_q$, reducing the number of edges from $|U_q| \times |W_q|$ to $|U_q| + |W_q|$ (Lines~18-23). 
In Line~22, $\mathcal{E}_{\mathcal{C}}$ is updated with edges from the extracted $\delta$-cliques.
Remaining edges in the bipartite graph $G(U,W, E)$, not part of any bipartite $\delta$-cliques, are connected directly in the tripartite graph (Line~23). In Line 24 the algorithm returns the restructured graph $G^*(U,Z, W, E^*)$ having a reduced number of edges. \looseness=-1

In Appendix~\ref{sec:appendix:example} we provide an example showing the execution of \textsf{CPGR}. 

\subsection{Clique Stripping  Algorithm \textsf{(CSA)}}
\label{subsec:S-CPGR}

The Clique Stripping Algorithm \textsf{(CSA)}, given in Algorithm~\ref{alg:CSA}, extracts $\delta$-cliques from the given bipartite graph~$G(U,W, E)$. \textsf{CSA} initially selects the vertices for the right partition~$K_c$ of the clique~$C_c$ and then the common neighbors for the set~$K_c$. The input of \textsf{CSA} consists of the updated index~$q$ of the clique, the size~$\hat{k}$ of the right partition  which guarantees a $\delta$-clique in~$G$, the adjacency matrix~$A$, the number of vertices~$n$, and the
degree $d_{w_j} =  |N(w_j)|$ . \looseness=-1

 \begin{algorithm}[!b]
\caption{{\small \textsf{CSA}: Clique Stripping Algorithm }}
{
\begin{algorithmic}[1]
\INPUT {$q$: Index of bipartite clique; }
\Statex	\hspace*{0.45cm}{$\hat{k}$: Size of the right bipartition in bipartite clique;}
\Statex	\hspace*{0.45cm}{$A$: Adjacency matrix; }
\Statex	\hspace*{0.45cm}{$n$: Number of vertices in left/right bipartition of $G$;}
\Statex	\hspace*{0.45cm}{$d_{w}$: Degree of vertices in bipartition $W$.}

\State{$ \mathcal{K} \gets \emptyset$}
\State{Sort $d_w$ in non-increasing order. Let $ d_{w_{\pi(1)}}, d_{w_{\pi(2)}},\ldots,d_{w_{\pi(n)}} $ be the order.}
\State{$j \gets 1$}
\While{$d_{w_{\pi(j)}} \geq d_{w_{\pi(\hat{k})}}$}
    \State{$ \mathcal{K} \gets \mathcal{K} \cup \{w_{\pi(j)}\}$}
    \State{$j \gets j + 1$}
\EndWhile
\State{$\gamma \gets   \big\lfloor{\frac{|\mathcal{K}|}{\hat{k}}} \big\rfloor$}
\For{ $ c = q + 1,\ldots,q + \gamma $}
    \State{$K_c \gets \{ w_{\pi_{(j)}} \in \mathcal{K} \ | \ (c - (q + 1) )\cdot \hat{k} < j \leq (c -q ) \cdot \hat{k} \} $}
\State{$U_{K_c} \gets \{ u_i \in U | K_c \subseteq N(u_i)\}$ }
\State{$C_c \gets \{ (U_{K_c}, K_c) \}$}
\State{Update $A$, by removing edges $( U_{K_c} \times  K_c) \in C_c$ }
\State{Update $d_w$ by subtracting $|U_{K_c}|$ from each $d_{w_{\pi(j)}}$, where $w_{\pi(j)} \in K_c$}
\EndFor
\State{$\mathscr{C} \gets \{C_{q+1},\ldots,C_{q + \gamma} \} $}

\State \textbf{Output} ($\mathscr{C}, \hat{A}, \hat{d_w}$)
\end{algorithmic}
\label{alg:CSA}
}
\end{algorithm}

\vspace*{0.07cm}
\noindent \textbf{Initialization} (Lines 1-3). \textsf{CSA} initializes the set~$\mathcal{K}$  of selected vertices for clique extraction to the empty set. 
In Line~2, it sorts~$d_w$ in non-increasing order of the degrees of vertices $w_j \in W$, i.e., $d_{w_{\pi(1)}} \geq d_{w_{\pi(2)}} \geq \ldots \geq d_{w_{\pi(n)}}$. In Line~3 it initializes the sorted vertices list index~$j$ to~1. \looseness=-1

\vspace*{0.07cm}
\noindent \textbf{Vertex selection} (Lines 4--7).
\textsf{CSA} selects vertices for clique extraction based on their degrees, since vertices with higher degree are more likely to share a larger set of common neighbors and thus participate in dense bipartite cliques. Rather than selecting exactly $\hat{k}$ vertices, \textsf{CSA} constructs a candidate set
$ \mathcal{K} = \{ w_{\pi(j)} \mid d_{w_{\pi(j)}} \ge d_{w_{\pi(\hat{k})}} \}$,
which contains all vertices whose degree is at least the $\hat{k}$-th largest degree in the sorted sequence~$d_w$. As a result, $|\mathcal{K}| \ge \hat{k}$ in general.
\textsf{CSA} adds these vertices to~$\mathcal{K}$ by iterating through the sorted order and incrementing the index~$j$ (Lines~5--6). Thus, the while-loop in Lines~4--6 ensures that all vertices tied at the $\hat{k}$-th degree threshold are included, avoiding arbitrary tie-breaking and forming a maximal candidate set for the right partition.

This threshold-based selection increases the likelihood of identifying a larger set of common neighbors in the left partition, thereby improving the quality of the extracted $\delta$-cliques. In Line~7, \textsf{CSA} computes $\gamma = \lfloor |\mathcal{K}| / \hat{k} \rfloor$, which determines how many $\delta$-cliques can be extracted. The set~$\mathcal{K}$ is then partitioned into $\gamma$ subsets of size~$\hat{k}$ (Line~9), allowing \textsf{CSA} to extract multiple $\delta$-cliques in a single execution.




\vspace*{0.07cm}
\noindent \textbf{Clique extraction} (Lines 8-13).
The for loop (Lines 8-9) first partitions~$\mathcal{K}$ into subsets~$K_c$, such that
$\bigcup_{c=q+1, \ldots, q+\gamma}K_c = \mathcal{K}$ 
and $K_{q+i}  \bigcap K_{q+j} = \emptyset, \ i = 1, \ldots, \gamma, j = 1, \ldots, \gamma, i \neq j.$
\label{eq:set_k_eq2}
Then it constructs $U_{K_c}$ such that each vertex in $U_{K_c}$ is part of the set of common neighbors of~$K_c$ (Line~10). In Line~11, \textsf{CSA} forms clique $C_c$ with partitions~$U_{K_c}$ and~$K_c$, where $|U_{K_c}| \geq \lceil{n^{1-\delta}\rceil} $ and $|K_c| = \hat{k}$. After forming the clique~$C_c$, in Line~12, the algorithm updates the adjacency matrix by removing the edges in the clique~$C_c$ and in Line~13, updates the degrees of the vertices that are part of the clique by removing the size of left partition from the degree of each vertex~$w_{\pi(j)} \in K_c$, i.e., $d_{w_{\pi(j)}} = d_{w_{\pi(j)}} - |U_{K_c}|, \  \forall \  w_{\pi(j)} \in K_c , \ c = q+1, \ldots q+\gamma$. Finally, in Line~14, \textsf{CSA} forms~$\mathscr{C}$, the set of all the cliques extracted in the current execution, and in Line~15, it returns the set of cliques~$\mathscr{C}$, the updated adjacency matrix~$\hat{A}$, and the updated degrees of vertices~$\hat{d}_w$ to~\textsf{CPGR}. \looseness=-1

\vspace*{0.07cm}
\noindent \textbf{Time complexity of \textsf{CSA}.} In Line~2 \textsf{CSA} takes $O(n \log n)$ to sort the degrees of vertices~$d_w$. The while loop in Lines 4-6, takes at most $O(n)$ to select vertices for the $\delta$-cliques. The for loop (Lines 8-13) executes~$O(\gamma)$ times, where~$1\leq \gamma \leq \frac{n}{\hat{k}}$. In the for loop, Line 9 and~10 takes $O(\hat{k})$ and $O(n\hat{k})$ to find the right and left partition of the $c$-th $\delta$-clique, respectively.  \textsf{CSA} takes $O(n\hat{k})$ to remove the edges in the $c$-th $\delta$-clique, in Line~12. It takes $O(\hat{k})$ to update the degrees $d_{w}\in K_{c}$, in Line~13. Therefore, the total running time of \textsf{CSA} is dominated by Lines 10 and 12 and thus  \textsf{CSA} takes $O(n \gamma \hat{k})$ time. \looseness=-1

\vspace*{0.07cm}
\noindent {\textbf{Time complexity of} \textbf{\textsf{CPGR}}}.  \textsf{CPGR} takes $O(n^2)$ to calculate the degrees~$d_w$ of vertices in~$W$ in Lines 6-8. The running time of the while loop (Lines 10-17) is dominated by the running time of the function \textsf{CSA} in Line~11 which is given by Algorithm~\ref{alg:CSA}. \textsf{CSA} takes $O(n\gamma \hat{k})$ to remove $\gamma \hat{k} n^{1-\delta}$ edges. Therefore, on average it takes $O(n^{\delta})$  time to remove one edge from the given graph. Thus, the while loop in \textsf{CPGR} takes $O(m n^{\delta})$ time to extract $\delta$-cliques. In the end, to restructure the graph, \textsf{CPGR} in Lines 18-21, takes linear time. Therefore, the total running time of \textsf{CPGR} is $O(mn^{\delta})$ which is dominated by the while loop in Lines 10-17.   \looseness=-1

\section{Properties of \textsf{CPGR}}
\label{properties}
The transformation performed by \textsf{CPGR} preserves vertex-to-vertex reachability and the connected components of the original graph. Consequently, the existence of paths between any pair of vertices in $G$ is maintained in the restructured graph $G^*$. However, \textsf{CPGR} does not necessarily preserve local structural properties such as vertex degrees, neighborhood structure, or exact shortest-path distances, since edges belonging to $\delta$-cliques are replaced by paths of length two.  
We formalize this preservation of path information in the following theorem.

\begin{theorem}
\label{theorem1}
The restructured graph~$G^*(U,W,Z,E^*)$ obtained by {\rm \textsf{CPGR}} preserves the path information (i.e., vertex-to-vertex reachability)  of the original graph~$G(U,W,E)$. 
\end{theorem}
\begin{proof} Let's assume that \textsf{CPGR} extracts only one $\delta$-clique $C_q$ from the given graph~$G$.  In  the $\delta$-clique~$C_q$, the right partition~$K_q\subseteq W$ is formed by selecting~$\hat{k}$ vertices from~$W$ and the left partition~$U_q \subseteq U$ is formed with the common neighbors of~$K_q$, forming a complete bipartite graph with edge set~$E_q$. The restructured graph~$G^*$ is formed with the same left and right partitions~$U$ and~$W$ of the graph~$G$, and a third partition~$Z$ which is a set of additional vertices associated with each of the $\delta$-cliques that \textsf{CPGR} extracts.  Our main concern is with the edges $(u_i,w_j) \in E^*$ in the restructured graph~$G^*$. $E^*$ contains the edges that replace the $\delta$-clique~$C_q$ by adding the additional vertex~$z_q \in Z$, and the set of edges $\hat{E} = E \setminus E_q$ which were not part of the $\delta$-clique~$C_q$. 
It can be easily seen that~$\hat{E} \subset E$, that is, the edges in~$\hat{E}$ are edges in~$G$, and the remaining edges $E^* \setminus \hat{E}$ in $E^*$ are the edges that replace~$C_q$ in~$G^*$. Each edge $(u_i,w_j) \in E_q$, is replaced by two edges, $(u_i,z_q)$ and $(z_q,w_j)$, where $u_i \in U$, $w_j \in W$, and $z_q \in Z$. Therefore, for each edge $(u_i,w_j) \in E_q$, there exists a path from~$u_i$ to~$w_j$ composed of two edges, $(u_i,z_q)$ and $(z_q,w_j)$, that passes through the additional vertex~$z_q \in Z$, thus preserving the path information. This holds true for all the $\delta$-cliques extracted by \textsf{CPGR}.
\end{proof}

In the following, we determine an upper bound on the number of edges in the restructured graph obtained by \textsf{CPGR}. First, we state a theorem from~\cite{federMotwani} that guarantees the existence of a $\delta$-clique in a bipartite graph. This theorem is necessary in the proof of the bound.

\begin{theorem}
    ${\mathrm{\cite{federMotwani}(Theorem \ 2.2)}}$ Every bipartite graph $G(U,W,E)$ contains a $\delta$-clique.
\end{theorem}


\begin{lemma}
\label{min_m_hat}
Given a graph $G(U,W,E)$, where $n = |U| = |W|$, $m = |E|$, and a constant $\delta$, $0 < \delta < 1$, if $m < 2n^{2-\frac{\delta}{2}}$ then extracting $\delta$-cliques and replacing them with tripartite graphs as done in {\normalfont \textsf{CPGR}} does not lead to a reduction in the number of edges of~$G$.
\end{lemma}
\begin{proof}
    In the $\delta$-clique based graph restructuring performed by \textsf{CPGR}, the restructured graph $G^*(U,W,Z,E^*)$ is obtained from~$G$ by replacing $|E_q| = |U_q \times K_q|$ edges in a $\delta$-clique~$C_q$ with $|E^*| = |U_q| + |K_q|$ edges in $G^*$ and adding an additional vertex $z_q \in Z$ for each extracted $\delta$-clique~$C_q$. 
    The size of the right partition of~$C_q$ is $|K_q|={k} = k(n,m, \delta)=\Big\lfloor{\frac{\delta \log n}{\log(2n^2/m)} }\Big\rfloor$. If $m < 2n^{2-\frac{\delta}{2}}$ then the size of the right partition~$|K_q|$ is $k = 1$. Therefore, the number of edges in the $\delta$-clique is $|E_q| = |U_q| \cdot 1 = |U_q|$. Those edges are replaced by $ |E^*|  = |U_q| + \hat{k} = |U_q| + 1$ edges in~$G^*$. Thus, replacing the $\delta$-clique in~$G$ with $m < 2n^{2-\frac{\delta}{2}}$, the number of edges in $G^*$ would actually increase by~1, which does not lead to a restructuring of~$G$. 
\end{proof}

\begin{theorem}
\label{comp_ratio}
    Let $G(U,W,E)$ be any bipartite graph with $|U| = |W| = n$ and $|E| = m > 2n^{2-\frac{\delta}{2}}$, where~$\delta$ is a constant such that $0 < \delta < 1$. Then, the number of edges in the restructured graph $G^*(U,W,Z,E^*)$ obtained by {\normalfont \textsf{CPGR}} is $|E^*| = O\left(\frac{m}{k(n,m,\delta)}\right)$.
\end{theorem}
\begin{proof}
    We follow the basic idea of the proof from Theorem 2.4 in \cite{federMotwani} and extend it to apply to the \textsf{CPGR} algorithm. We assume that initially the restructured graph $G^*$ has 0 edges, i.e.,~$|E^*| = 0$, and that edges are added to~$E^*$ as the algorithm progress. To estimate the number of edges in~$E^*$, we divide the iterations of \textsf{CPGR} into stages, where each stage~$i$ consists of extracting one or more~$\delta$-cliques with a fixed~$\hat{k}_i$. Therefore, the~$i^{th}$ stage includes all $\delta$-cliques extracted after the number of edges in~$G$ becomes less than~$m/2^{i-1}$ for the first time, and before the number of edges in~$G$ becomes less than~$m/2^{i}$ for the first time. For stage~$i$,  $\hat{k}$ is always going to be at least
     $   k_i =  \frac{\delta \log n}{2 \log(2^i \cdot 2n^2/m)}$.
    
    Assume that a clique~$C_q$ with right partition~$K_q \subseteq W$ and left partition~$U_{K_q} \subseteq U$ is extracted in stage~$i$. Extracting~$C_q$ removes $|K_q \times U_{K_q}|$ edges from $G$ and adds $|K_q| + |U_{K_q}|$ edges in~$E^*$. Therefore, the average number of edges added in~$E^*$ for each edge removed from~$G$ by extracting~$C_q$ is $\rho = \frac{|K_q|+|U_{K_q}|}{|K_q \times U_{K_q}|} = \frac{1}{|K_q|} + \frac{1}{|U_{K_q}|}$. From the definition of a $\delta$-clique, during stage~$i$ we have $|K_q| = \hat{k} \geq k_i$ and $|U_{K_q}| \geq n^{1-\delta}$, and therefore, $\rho \leq \frac{1}{k_i} + \frac{1}{n^{1-\delta}}$.
    The total number of edges removed from~$G$ in stage~$i$ cannot be greater than $2m/2^i$, therefore, the number of edges added to~$E^*$ during stage~$i$ is less than or equal to $\rho \frac{2m}{2^i} = \left(\frac{1}{k_i} + \frac{1}{n^{1-\delta}}\right)\frac{2m}{2^i}$.  

    \textsf{CPGR} terminates extracting cliques when the number of remaining edges in~$G$ is less than~$2n^{2 - \frac{\delta}{2}}$. This is to eliminate the extraction of trivial cliques, as shown in Lemma~\ref{min_m_hat}.
    Next, we determine an upper bound on the total number of edges~$m^*$ added by \textsf{CPGR} to~$G^*$ before it terminates extracting $\delta$-cliques (i.e., when $m < 2n^{2 - \frac{\delta}{2}}$). 
    Therefore, 
    \begin{gather}
    \begin{split}
        m^* \leq \sum_{i=1}^{\big\lceil\log\left(\frac{m}{2n^{2 - \delta/2}}\right)\big\rceil} \bigg(\frac{1}{k_i} + \frac{1}{n^{1-\delta}}\bigg)\frac{2m}{2^i} 
        \leq \sum_{i=1}^{\infty} \bigg(\frac{1}{k_i} + \frac{1}{n^{1-\delta}}\bigg)\frac{2m}{2^i}.
    \end{split}
\end{gather}

    Since $1/(2k_i) \leq 1/k + i/(\delta\log n)$, we obtain,
    \begin{gather}
    \begin{split}
         m^* \leq \frac{4m}{k} \sum_{i=1}^{\infty} \frac{1}{2^i} + \frac{4m}{\delta \log n} \sum_{i=1}^{\infty} \frac{i}{2^i} + \frac{2m}{n^{1-\delta}} \sum_{i=1}^{\infty} \frac{1}{2^i} 
         \leq 2m \left(\frac{2}{k} + \frac{4}{\delta \log n} + \frac{1}{n^{1-\delta}}\right).
    \end{split}
\end{gather}
    Thus, $m^* = O\left(\frac{m}{k}\right)$. After \textsf{CPGR} finishes extracting cliques, there are $2n^{2-\delta/2}$ edges remaining in~$G$. Those remaining edges are trivial cliques (i.e., single edges) that are added to~$G^*$. The number of remaining edges $2n^{2-\delta/2}$ is in $O\left(\frac{m}{k}\right)$. Therefore, the number of edges in the restructured graph $G^*(U,W,Z,E^*)$ obtained by \textsf{CPGR} is $|E^*| = O\left(\frac{m}{k}\right)$.
\end{proof}

\vspace*{0.02cm}
\noindent {\textbf{Extension to non-bipartite graphs.}} 
\textsf{CPGR} can be extended to restructure non-bipartite graphs using a similar technique to that described in~\cite{federMotwani}. Consider a directed graph $H(V, E)$ with $n$ vertices and~$m$ edges. Initially, we transform $H$ into a bipartite graph $G\left(L, R, E^{\prime}\right)$, where each vertex $v \in V$ is duplicated into $v_L \in L$ and $v_R \in R$. For each directed edge $(u, v)$ in $E$, we create a directed edge $\left(u_L, v_R\right)$ in~$E^{\prime}$, with the direction from~$L$ to~$R$. Furthermore, for each vertex $v \in V$ we add a directed edge from~$v_R$ to~$v_L$. The proposed graph restructuring algorithm is then applied to the bipartite graph~$G$, which includes only the edges from~$L$ to~$R$. Once the restructured  graph is computed we add to it the~$n$ directed edges from~$u_R$ to~$u_L$, corresponding to each vertex $u\in V$. 
Since the path information of the original graph is preserved by this transformation, it allows different graph algorithms to work on the restructured graph as well. 
Undirected general graphs can be transformed first into directed graphs by simply replacing each undirected edge $(u,v)$ in the original graph by two directed edges $(u,v)$ and $(v,u)$, and then applying the technique presented above to transform the directed graph into a bipartite graph. 
The time complexity and the compression ratio of the proposed algorithms are the same as in the case of the bipartite graphs. 
The time complexity and the reduced number of edges of \textsf{CPGR} are the same as in the case of the bipartite graphs. A similar approach for converting a general graph to a bipartite graph and then restructuring the obtained bipartite graph was employed in~\cite{1scalable,compressing_bisection}.

\section{Experimental Results}
\label{sec: experimental results}

We evaluated the running time and the \emph{reduction in the number of edges}, $R$ of \textsf{CPGR} by comparing it with the \textsf{FM} algorithm~\cite{federMotwani}.
Let $m = |E|$ be the number of edges in the given graph $G$ and $m^* = |E^*|$ be the number of edges in the restructured graph $G^*$ and $m^* \leq m$, then we define the reduction in the number of edge as $R = m - m^*$.
Due to the lack of an existing implementation, we implemented \textsf{FM} in C. However, its vertex selection calculations encountered machine representation limits on large graphs, restricting our comparative analysis to small bipartite graphs (up to $n = 128$ vertices per partition, $\sim$16k edges). To ensure a fair runtime comparison, we excluded trivial clique extraction ($\hat{k} = 1$) for both algorithms (\textsf{CPGR} inherently does this as per Algorithm~\ref{alg:CPGR}, Line~10). 
We further assessed \textsf{CPGR} on large bipartite graphs (up to 1.05 billion edges) and real-world datasets from the SuiteSparse Matrix Collection~\cite{SuiteSparse}. Additionally, we tested the utility of our restructured graph by applying the Dinitz algorithm for maximum bipartite matching to large bipartite graphs and the APSP algorithm on real-world datasets.
All algorithms, 
were implemented in C  and tested on a Linux system (AMD EPYC-74F3, 3.2GHz, 1 core, 128 GB RAM) using GCC 8.5.0. The implementation is available at~GitHub~\cite{github}.


Our experiments involved synthetic bipartite graphs generated using adapted $G(n,p)$ random graph model with $|U| = |W| = n$, $n \in [32, 32768]$, densities $p \in [0.8, 0.98]$, up to 1.05 billion edges and general graphs with standard $G(n,p)$ with $n = 32000$, densities $p \in [0.4, 0.7]$, $\sim$409 to 716 million edges, with 10 instances per configuration. We also used real-world sparse graphs from SuiteSparse~\cite{SuiteSparse}. For each instance, we ran \textsf{CPGR} with six $\delta$ values [0.5 to 1] and show average results and standard deviations for runtime and reduction in number of edges $R$, as well as the performance of the application algorithms on the restructured graphs. \looseness -1

\subsection{Results for Small Bipartite Graphs}

\begin{figure}[h]
     \centerline{
    \includegraphics[width=\linewidth]{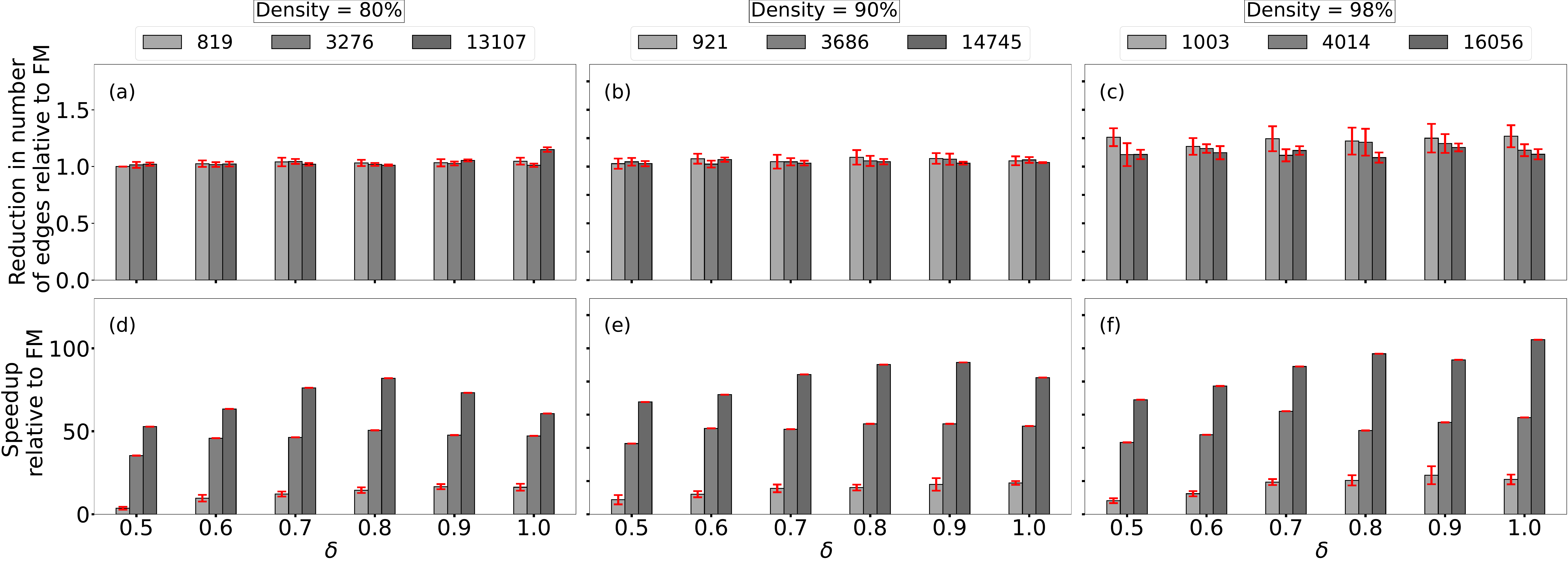}
    \label{fig:relative_comparision}}
  \caption{\textsf{CPGR} vs.\ \textsf{FM}: Reduction in number of edges  relative to \textsf{FM} (top row: (a), (b) and (c)) and speedup relative to \textsf{FM} (bottom row: (d), (e) and (f)) for graphs with 819 to 16~thousand edges.} 
  \label{fig:relative_CR_speedup}
\end{figure}

In this section, we compare the performance of \textsf{CPGR} against \textsf{FM} on small bipartite graphs, particularly, with~$n = 2^i$ vertices in each bipartition, where~$i$ ranges from~5 to~7 and the number of edges ranging from 819 to 16,056. However, due to the limitations of \textsf{FM} as previously mentioned, our comparison is constrained to $n = 2^7$ and $m = 16,056$. We compare the performance of the two algorithms using two metrics, speedup, and  $R$ relative to \textsf{FM}. \looseness=-1

\vspace*{0.1cm}
\noindent {\bf Reduction in number of edges $R$ relative to \textsf{FM}.} We define the \emph{edge reduction ratio relative to} \textsf{FM} as the ratio of the $R$ achieved by \textsf{CPGR} to that of \textsf{FM}.
Figures~\ref{fig:relative_CR_speedup}a-\ref{fig:relative_CR_speedup}c illustrate the relative~$R$ of \textsf{CPGR} compared to \textsf{FM} for small bipartite graphs across different $\delta$ values (0.5 to 1) and densities (0.80, 0.90, 0.98). In all these cases, the relative $R$ ranges from 1 to 1.27, indicating that \textsf{CPGR} achieves at least the same, and often a greater reduction than \textsf{FM}. This improved performance of \textsf{CPGR} stems from its ability to extract multiple $\delta$-cliques before  $\hat{k}$ decreases (as shown in Figure~\ref{fig:motivation}b), allowing for the removal of more edges at higher $\hat{k}$ values (Section~\ref{subsec:S-CPGR}).
Notably, for a graph with $n=32$, $m=819$, a density of 0.80, and $\delta$=0.5, \textsf{CPGR} achieves an $R$ equivalent to \textsf{FM}. This is because the small size of the graph limits the number of iterations required for $\delta$-clique extraction, thus not providing \textsf{CPGR} with a significant advantage over \textsf{FM}. However, for the same $n=32$ and $\delta$, increasing the graph density leads to a higher $R$ for \textsf{CPGR}. For instance, when $m$ increases from~921 to~1,003, the relative $R$ rises from 1.03 to 1.25. In contrast, for a graph $n=128$, $m=16056$, density = 0.98, we observe a non-monotonic trend in the relative $R$ as $\delta$ increases. Initially, the ratio increases from~1.1 to~1.14 as $\delta$ goes from 0.5 to 0.7. However, further increasing~$\delta$ to~1 results in a decrease to~1.1. This reduction is to the fact that a larger $\delta$ increases the size of the right partition $K_q$ of a $\delta$-clique (since $|K_q| = \hat{k}$ and $\hat{k} \propto \delta$),
consequently decreasing the likelihood of finding common neighbors. \looseness=-1

It is important to recognize that the optimal $R$ is influenced by the interplay between $n, m,$ and~$\delta$. Identifying the specific combinations of these parameters that yield the highest reduction in the number of edges~$R$ remains a subject for future investigation.
\looseness=-1

\vspace*{0.1cm}
\noindent {\bf Speedup relative to \textsf{FM}.} 
We define the \emph{speedup relative} to \textsf{FM} as the ratio of the runing time of \textsf{FM} over that of \textsf{CPGR}. Figures~\ref{fig:relative_CR_speedup}d-f illustrate the speedup achieved by \textsf{CPGR} for bipartite graphs with $n$ = 32, 64, and 128 vertices in each partition, $m$ = 819 to 16 thousand for densities of 0.80, 0.90, and 0.98, and $\delta$ values ranging from 0.5 to 1. The running time of \textsf{CPGR} is consistently lower than that of \textsf{FM}, resulting in speedup by \textsf{CPGR} across all cases. This is primarily due to the selection of vertices in the while loop (Lines 4-6) in \textsf{CSA}, facilitating greater~$R$ by extracting more than one (i.e., $\gamma$), $\delta$-cliques in a single iteration, as shown in Figure~\ref{fig:motivation}b. The speedup increases notably with higher $\delta$ and density. For example, with $n = 32$ and $m = 819$ for density 0.8, \textsf{CPGR} achieves an average speedup of~3.66 for $\delta = 0.5$, increasing to~16.28 for $\delta = 1$. Similarly, with $n = 32$ and $\delta = 1$, \textsf{CPGR} achieves an average speedup ranging from 16.28 for $m = 819$ with density 0.8 to 20.97 for $m = 16$ thousand with density 0.98. It is worth noting that for small graphs, certain combinations of $n$, $m$, and $\delta$ yield higher speedup values. For instance, with $n = 128$ and $m = 819$ for density 0.8, the speedup increases from 52.81 to 81.97 for $\delta = 0.8$, and then slightly decreases to 60.68 for $\delta = 1$.
\looseness=-1

\subsection{Results for Large Bipartite Graphs}

We present the results obtained from testing \textsf{CPGR} on large bipartite graphs, where each bipartition consists of~$n=2^i$ vertices, with~$i$ ranging from~11 to~15, and a number of edges~$m$ ranging from approximately 3.36~million to 1.05~billion.   \looseness=-1

\begin{figure}[!tb]
  \centerline{
  \includegraphics[width=\textwidth]{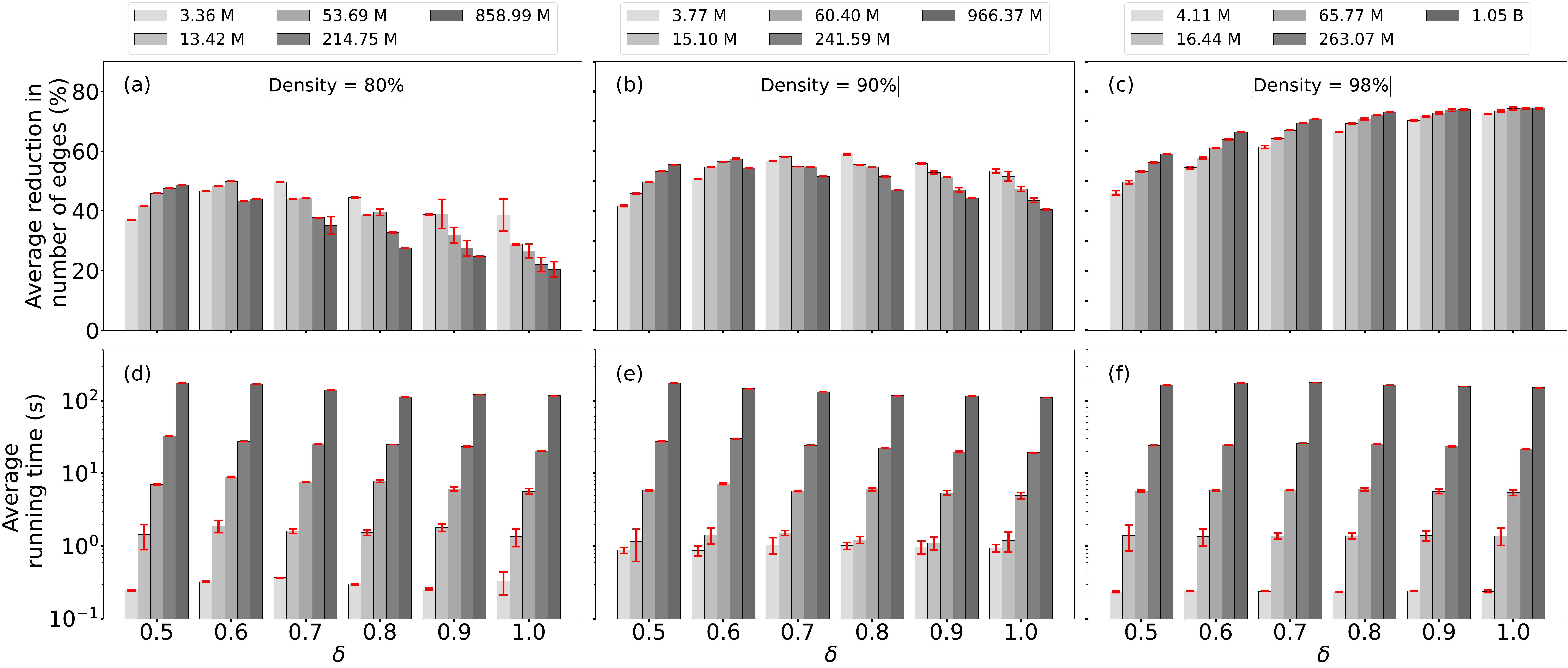}
  }
  \caption{\textsf{CPGR on Large Graphs}: Average reduction in the number of edges $R^{avg}$ ((a), (b), and (c)) and average running time ((d), (e), and (f)) for graphs with 3.36~million to 1.05~billion edges. \looseness -1} 
 \label{fig:combined_CPGR_results}
\end{figure}

\vspace*{0.1cm}
\noindent {\bf Average reduction in the number of edges ($R^{avg}$).} Figures~\ref{fig:combined_CPGR_results}a-\ref{fig:combined_CPGR_results}c illustrate $R^{avg}$  for \textsf{CPGR}, defined as the percentage of edges removed from the original graph to obtain the restructured graph, i.e., $R^{avg}=\frac{m - m^*}{m} \cdot 100$.
We observe a consistent trend: for fixed $n$ and $\delta$, $R^{avg}$ improves with increasing density. For example, when $m = 214.75$ million edges at a density of $0.80$ and $\delta = 0.5$,  $R^{avg}$  is $48.72\%$, whereas $R^{avg}$  rises to $59\%$ for $m = 1.05$ billion edges at a density of~$0.98$. This correlation arises because higher density increases the likelihood of finding a large set of common neighbors for the right partition of a $\delta$-clique. Identifying more common neighbors allows \textsf{CPGR} to remove more edges, thereby improving $R^{avg}$. \looseness -1

However, when the density is held constant, $R^{avg}$  exhibits a non-monotonic behavior depending on $\delta$ and $n$. Consider Figure~\ref{fig:combined_CPGR_results}a, for $m = 53.69$ million edges at density $0.80$, $R^{avg}$  initially increases from $45.6\%$ at $\delta = 0.5$ to a peak of $50\%$ at $\delta = 0.6$. Further increasing $\delta$ to $1$, however, leads to a decrease in $R^{avg}$  to $26.5\%$. 
As explained before, this is because a larger $\delta$ increases the size of the right partition~$K_q$ of a $\delta$-clique, which consequently lowers the possibility of finding common neighbors. \looseness -1 

\vspace*{0.1cm}
\noindent {\bf Average running time.} 
For a fixed $\delta$ and density, the running time of \textsf{CPGR} increases with the number of vertices, consistent with its $O(mn^{\delta})$ complexity. Conversely, for a fixed $n$ and density, the running time decreases as $\delta$ increases. For example, at density $0.80$ ($m = 214.75$ million), the runtime drops from $176.27$ s at $\delta = 0.5$ to $117.8$ s at $\delta = 1$. This is because a larger $\delta$ limits the search for common neighbors, leading to lower reduction in the number of edges at lower densities and allowing \textsf{CPGR} to skip more vertices. Consequently, the expected increase in runtime with $\delta$ is not observed. Notably, even at a high density of $0.98$ ($m = 1.05$ billion), the running time still decreases with increasing $\delta$ (from $164.83$ s at $\delta = 0.5$ to $150.97$ s at $\delta = 1$). This is attributed to the extraction of larger $\delta$-cliques at higher densities and $\delta$, which reduces the overall number of iterations in \textsf{CPGR}, thus lowering the runtime. \looseness=-1

\subsection{Results for General Graphs}
\label{subsec:general_graphs}


\begin{figure}[t]
\centering
    \includegraphics[width=\textwidth]{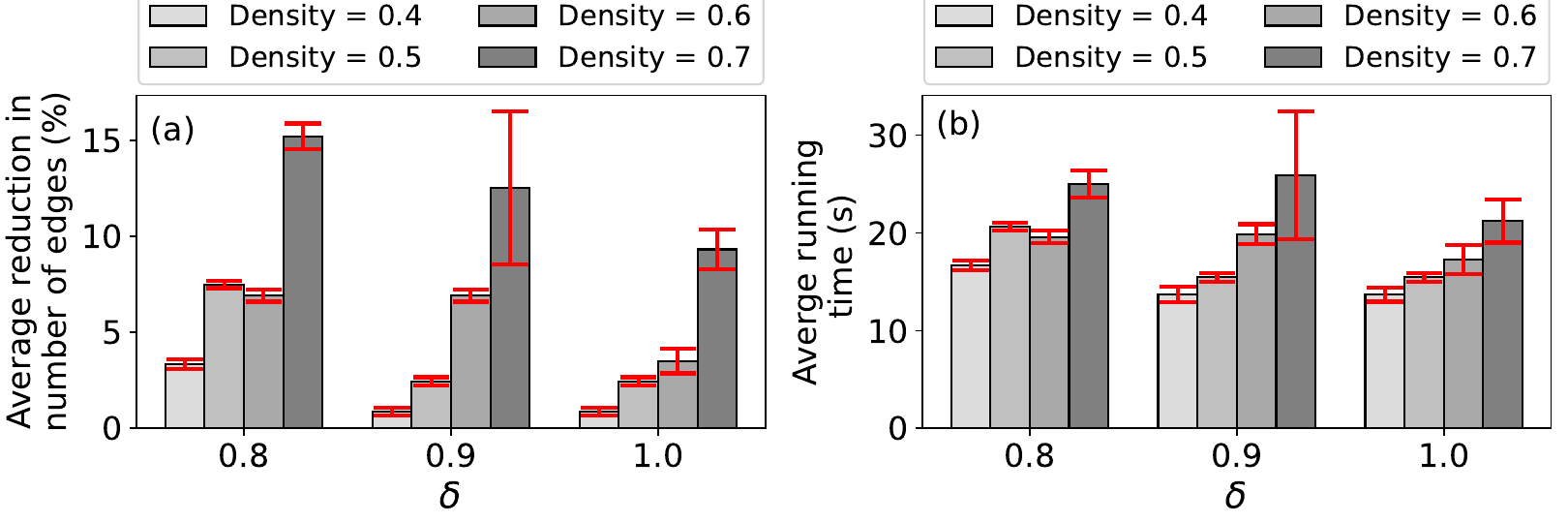}%
\caption{\textsf{CPGR on General Graphs}: (a)  average reduction in number of edges $R^{avg}$ and (b) average running time  for general graphs with 409~million to 716~million edges.}
\label{fig:general_graphs_results}
\end{figure}

In this section, we analyze \textsf{CPGR}'s performance on synthetically generated general graphs with~$n=32000$ vertices and number of edges~$m$ ranging from approximately 409~million to 716~million as explained in Section~\ref{sec: experimental results}. 
We observe that $R^{avg}$ and running time follow the same trend with varying density and $\delta$ as seen in the large bipartite graphs.  \looseness=-1

\vspace*{0.1cm}
\noindent {\bf Average reduction in the number of edges ($R^{avg}$).} Figure~\ref{fig:general_graphs_results}a shows that for a fixed~$\delta$, $R^{avg}$ improves with density. For instance, at $\delta = 0.6$, $R^{avg}$ increases from $12.3\%$ for $m \approx 409$ million at density $0.4$ to $28.6\%$ for $m \approx 716$ million at density $0.7$, consistent with the increased likelihood of finding common neighbors. Similar to large bipartite graphs, $R^{avg}$ is non-monotonic with $\delta$ at a constant density. For $m \approx 716$ million at density $0.7$, $R^{avg}$ peaks at $33.3\%$ for $\delta = 0.6$, after initially increasing from $28.6\%$ at $\delta = 0.5$, before dropping to $9.1\%$ at $\delta = 1$.  \looseness=-1

\noindent {\bf Average running time.}
Running time exhibits a non-monotonic trend with increasing~$\delta$ and density. Larger $\delta$ leads to larger right partitions, reducing the likelihood of finding common neighbors in low-density graphs, thus fewer extracted cliques with larger partitions. Conversely, low $\delta$ (e.g., 0.5) results in smaller right partitions, increasing the chance of finding common neighbors, leading to higher $R$ and longer runtime, as evident in Figure~\ref{fig:general_graphs_results}b. The highest runtime (59.96 s) corresponds to the highest $R^{avg}$ ($33.3\%$) at density $0.7$ and $\delta = 0.6$, while the runtime reduces to~21.21 s when~$\delta = 1$, and $R^{avg}$ drops to $9.1\%$ at the same density.  \looseness=-1

\begin{table}[b]
\centering
\caption{Real-world Datasets (SuiteSparse Matrix Collection \cite{SuiteSparse})} 
\label{tab:dataset_stats} 

\begin{tabular}{c l r r r r}
\toprule
Sr. No & Dataset name & Rows & Columns & Edges & Density \\
\midrule
1 & \textsf{beacxc} & 497 & 506 & 50409 & 0.2 \\
2 & \textsf{beaflw} & 497 & 507 & 53403 & 0.2 \\
3 & \textsf{cari} & 400 & 1200 & 152800 & 0.3 \\
4 & \textsf{mbeacxc} & 496 & 496 & 49920 & 0.2 \\
5 & \textsf{mbeaflw} & 496 & 496 & 49920 & 0.2 \\
\bottomrule
\end{tabular}

\end{table}

\begin{figure*}[t]
\centering
    \includegraphics[width=\textwidth]{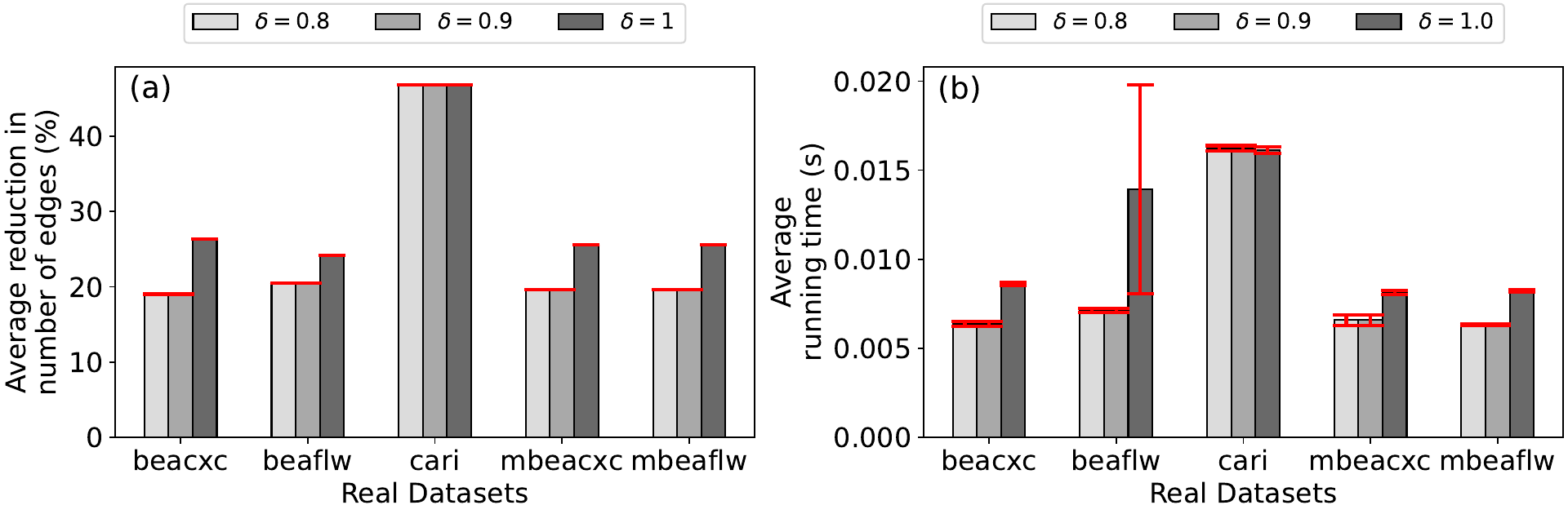}%
\caption{\textsf{CPGR on Real-world graphs}: (a) average   reduction in the number of edges ($R^{avg}$)  and (b) average running time.}
\label{fig:real_graphs_results}
\end{figure*}

\subsection{Results for Real-world Graphs}
\label{sec:real_world_results}
In this section, we analyze \textsf{CPGR} performance on low-density (up to 0.3), symmetric and asymmetric real-world graphs shown in Table~\ref{tab:dataset_stats}. Given that $\delta < 0.8$ would yield $\hat{k} = 1$ and thus no reduction in the number of edges for these sparse graphs, we tested \textsf{CPGR} with $\delta \in \{0.8, 0.9, 1\}$.
\looseness=-1

\noindent {\bf Average reduction in the number of edges ($R^{avg}$).} 
Figure~\ref{fig:real_graphs_results}a demonstrates that \textsf{CPGR} achieves $R^{avg}$ ranging from $18.7\%$ to $46.8\%$ even for the relatively low-density real-world graphs examined. Notably, for dataset \textsf{cari}, $R^{avg}$ remains consistent across all tested $\delta$ values. This stability is due to low density, which limits the maximum value of $\hat{k}$ to 2, thereby resulting in same the $R^{avg}$ regardless of $\delta$. While for other real-world graphs, $R^{avg}$ increases for $\delta$ = 1, this is because, when $\delta=1$, $\hat{k}$ larger for longer iterations of \textsf{CPGR} (as shown in Figure~\ref{fig:motivation}b), resulting in increase in $R^{avg}$. \looseness=-1

\noindent {\bf Average running time.}
The running time of \textsf{CPGR} on the real-world datasets is presented in Figure~\ref{fig:real_graphs_results}b. Similar to $R^{avg}$, the running time generally remains stable for most datasets, which is due to the same reason of low density limiting $\hat{k}$. However, we observe for the datasets that show increase in $R^{avg}$ at $\delta$ = 1, has increase in running time, For example, the average running time for dataset \textsf{beaflw} increases from 0.007 s at $\delta = 0.8$ and $0.9$ to 0.014 s at $\delta = 1$. This increase in running time is because at $\delta = 1$, \textsf{CPGR} could extract larger $\delta$-cliques (with $\hat{k} = 2$), leading to higher $R^{avg}$ and thus a longer running time. 
\looseness=-1

\vspace{-2mm}
\subsection{Speeding-up Graph Algorithms}
In this section, we show the speedup achieved by downstream graph algorithms when we use \textsf{CPGR} algorithm as a preprocessing step across different graph types. Specifically, we present the speedup of the APSP algorithm on real-world datasets and the speedup of Dinitz's algorithm on large synthetic bipartite datasets. Across these algorithms and dataset types, utilizing graphs with reduced number of edges obtained by \textsf{CPGR} as input consistently results in smaller running times. \looseness -1

\begin{figure}[!t]
  \centerline{
  \includegraphics[width=\linewidth]{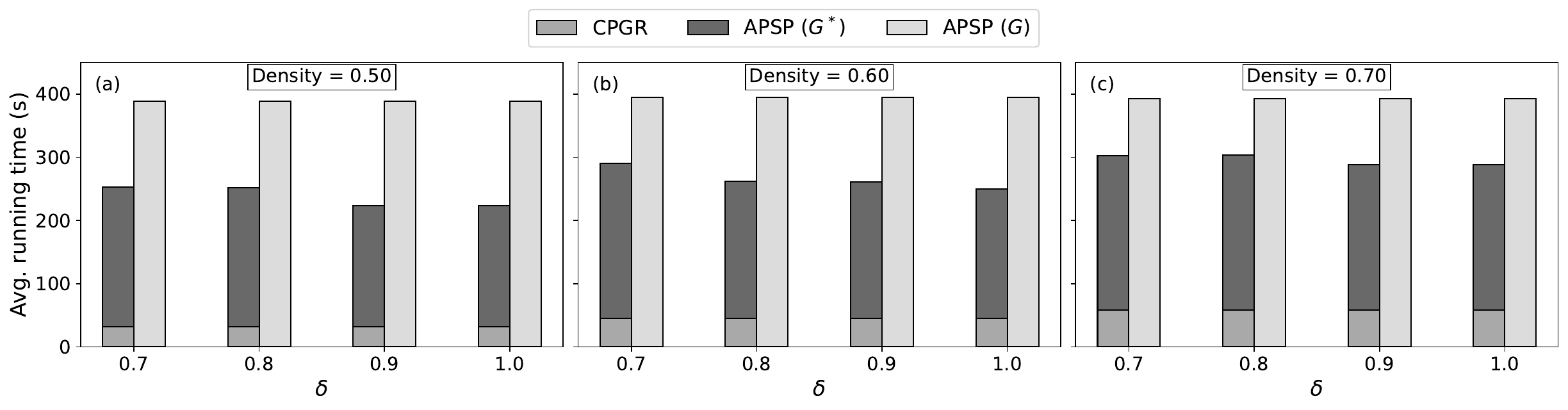}
  }
  \caption{\textsf{CPGR} + \textsf{APSP}($G^*$) vs.\ \textsf{APSP}($G$): Average running time of APSP algorithm on the original graph $G$ (labeled \textsf{APSP}($G$)) and on the restructured graph $G^*$ (labeled \textsf{APSP}($G^*$)), for large graphs with approximately 32,000 vertices and 256 million to 360 million edges corresponding to densities = $0.50, 0.60,$ and $0.70$ and different~$\delta$.}
 \label{fig:apsp_32000}
 \vspace*{-0.5cm}
\end{figure}

\vspace*{0.1cm}
\noindent \textbf{Speeding-up All-Pairs-Shortest-Path Algorithms.}
We evaluate the effectiveness of \textsf{CPGR} for speeding up Breadth-First Search (BFS)–based All-Pairs Shortest Path (\textsf{APSP}) computations on unweighted general graphs.
Figure~\ref{fig:apsp_32000} shows the running time of \textsf{APSP} on the original graph $G$ and on the restructured graph $G^*$ obtained via \textsf{CPGR}.
For each density, the bars represents the average cost of \textsf{CPGR}, \textsf{APSP}($G^*$), and \textsf{APSP}($G$), measured over 10 runs.
To validate correctness, we compare the number of connected vertices in the original and restructured graphs.
This validation confirms that vertex-to-vertex reachability is preserved between $G$ and $G^*$. Moreover, any path discovered by BFS in the restructured graph $G^*$ can be mapped back to a corresponding path in the original graph $G$ by contracting the length-two paths introduced by \textsf{CPGR} for $\delta$-cliques. Therefore, while path lengths in $G^*$ may differ, the \textsf{APSP} results remain correct with respect to the original graph.

Figure~\ref{fig:apsp_32000} shows the results on synthetic graphs with $n=32000$ vertices, varying both graph density and the parameter $\delta$.
The results show that, across all tested densities, the combined execution time of \textsf{CPGR} and \textsf{APSP}($G^*$) consistently outperforms \textsf{APSP}($G$) on large graphs.
Although the preprocessing cost increases with density, the reduction in \textsf{APSP} runtime offsets this overhead, yielding speedups of up to 1.74x.
These results demonstrate that \textsf{CPGR} remains effective even for large, moderately dense graphs.\looseness -1

\vspace*{0.1cm}
\noindent \textbf{Speeding-up Matching Algorithms.}
\begin{figure}
  \centerline{
  \includegraphics[width=\linewidth]{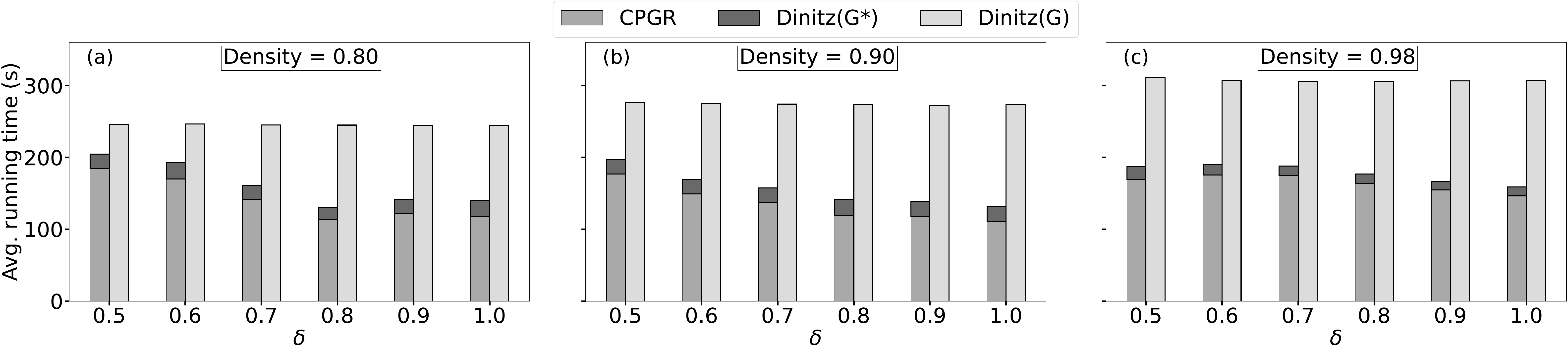}
  }
  \caption{\textsf{CPGR} + \textsf{Dinitz}($G^*$) vs.\ \textsf{Dinitz}($G$): Average running time of \textsf{Dinitz}'s algorithm on the original bipartite graph $G$ (labeled \textsf{Dinitz}($G$)) and on the restructured graph $G^*$ (labeled \textsf{Dinitz}($G^*$)), for a large graph with approximately 32,000 vertices in each bipartition and 819.75 million  to 1.05 billion edges corresponding to densities = $0.80, 0.90,$ and $0.98$ and different~$\delta$.}
 \label{fig:avg_execution_time}
 \vspace*{-0.5cm}
\end{figure}
The restructured graph retains path information, enabling its use as input for other graph algorithms like maximum cardinality matching to accelerate execution. 
For bipartite matching, we focus on \textsf{Dinitz}’s algorithm, which computes a maximum matching via the standard unit-capacity flow formulation. In this formulation, each edge has unit capacity and feasible flows correspond exactly to matching. Under the \textsf{CPGR} restructuring, each edge $(u_i, w_j)$ belonging to an extracted $\delta$-clique is replaced by a two-edge path $(u_i, z_q, w_j)$ in the restructured graph $G^*$, where both edges have unit capacity. This transformation preserves the set of feasible flows: a flow can traverse $(u_i, z_q, w_j)$ in $G^*$ if and only if the corresponding edge $(u_i, w_j)$ could be used in $G$, and matching constraints continue to be enforced by the unit-capacity edges incident to vertices in the left and right partitions. Consequently, \textsf{Dinitz}’s algorithm computes a maximum flow of the same value on $G^*$ as on $G$, and the resulting matching can be mapped back to a maximum matching in $G$ in linear time.

\textsf{Dinitz}'s algorithm~\cite{Dinitz:dans70}, with a runtime of $O(\sqrt{n} m)$ for bipartite matching, benefits from \textsf{CPGR} preprocessing. Given \textsf{CPGR}'s reduction in the number of edges $k=\Omega\Big(\frac{\delta \log n}{\log(2n^2/{m})}\Big)$ (matching \textsf{FM}'s), the total time becomes $O(\sqrt{n} m /k + mn^{\delta})$, which simplifies to $O(\sqrt{n}m \log_n{(\frac{n^2}{m}}))$ for $\delta < 1/2$. Notably, this matches the asymptotic bound for maximum cardinality bipartite matching, achieved by both \textsf{FM} and push-relabel methods. Similar speedups are expected with the Hopcroft-Karp algorithm~\cite{Hopcroft:sjc73}. 
Better asymptotic bounds (almost linear in $m$) are achieved by recent bipartite cardinality matching algorithms such as those presented in~\cite{Chuzhoy:stoc24} ($n^{2+o(1)}$-time randomized algorithm) and~\cite{Brand:focs20} ($\tilde O(m + n^{1.5})$-time randomized algorithm). These algorithms are mainly of theoretical interest, providing the best bounds, and not of practical interest. This is due to the use of sophisticated techniques that require significant implementation effort and they might not even be feasible to implement in practice. Thus, for our experimental analysis we use \textsf{Dinitz}'s algorithm as the baseline for comparisons.

Figure~\ref{fig:avg_execution_time} compares the running time of \textsf{Dinitz}'s algorithm on the original bipartite graph $G$ ($n=32$k per partition, $m=214.75$ Million to 1.05 Billion edges, densities 0.80-0.98) with the combined time of \textsf{CPGR} preprocessing and \textsf{Dinitz} on the restructured graph $G^*$. In all cases, the combined time outperforms \textsf{Dinitz} on $G$, achieving up to a 51.63\% reduction (2.07x speedup). 
On average \textsf{CPGR} + \textsf{Dinitz} ($G^*$) achieves 1.69x speedup when compared to \textsf{Dinitz} ($G$). The speedup correlates with the reduction in the number of edges; for example, at density 0.98, the increasing reduction in the number of edges by \textsf{CPGR} with higher $m$ and $\delta$ leads to a faster \textsf{Dinitz} on $G^*$, resulting in up to a 96.07\% reduction in \textsf{Dinitz}'s runtime compared to $G$. \looseness -1

\vspace{-2mm}
\section{Conclusion and Future Work}
\label{conclusion}
We proposed a novel Clique Partition based Graph Restructuring (\textsf{CPGR}) algorithm that operates by partitioning the graph into bipartite cliques. The algorithm has a time complexity of $O(mn^\delta)$. Experimental results demonstrate the effectiveness of \textsf{CPGR}, achieving a reduction in the number of edges by up to 21\% greater and execution times up to~105.18 times faster than the baseline \textsf{FM} algorithm. Furthermore, we investigated the impact of using the \textsf{CPGR}-restructured graph as input for a cardinality matching algorithm (\textsf{Dinitz}'s algorithm) and an All-Pairs Shortest Path (\textsf{APSP}) algorithm. Experimental results show that for sufficiently large dense graphs, our approach leads to a reduction in the total running time (including \textsf{CPGR} time) of up to~51.63\%, resulting in a speedup of 2.07x over the matching algorithm on the original graph. Similarly, for the \textsf{APSP} algorithm, we observed up to 34.78\% reduction in total running time, achieving a speedup of up to 1.74x. These speedups for fundamental graph algorithms highlight the practical benefits of our graph  restructuring technique, particularly for large-scale real-world applications like route planning and social network analysis where faster processing is critical for efficiency and responsiveness. \looseness -1

Looking ahead, our results indicate that the edge reduction achieved by \textsf{CPGR} is influenced by the interaction between the graph size, the graph density, and the parameter $\delta$. A deeper understanding of these interactions, including the  behavior observed for certain dense graphs, remains an open problem. Future work will focus on theoretically characterizing these interactions, developing adaptive strategies for selecting $\delta$, and extending \textsf{CPGR} to dynamic and streaming graph settings.
\looseness -1

\bibliography{references}
\bibliographystyle{abbrv}

\appendix
\newpage
\appendix
\section{Appendix}
\label{sec:appendix}

\tikzstyle{vertex}=[circle,draw,minimum size=14pt,inner sep=0pt]
\tikzstyle{edge} = [draw,thick,-]
\tikzstyle{weight} = [font=\small]
\usetikzlibrary{decorations.markings}

\usetikzlibrary{shapes,decorations,arrows,calc,arrows.meta,fit,positioning, shapes.geometric}
\tikzset{edge/.style = {->,> = latex'}}

\begin{figure*}[b!]
\centerline{
    \begin{tikzpicture}[ball/.style={ellipse,  draw}, >=LaTeX]
    \tikzstyle{vertex}=[circle,draw,minimum size=8pt,inner sep=0pt]
    \tikzstyle{vertex_set} = [ball/.style={ellipse, minimum width=0.5cm, minimum height=0.5cm, draw}, >=LaTeX]
        \tikzset{edge/.style = {->,> = latex'}}
        \foreach \pos/\name in {
          {(0,1)/w_1}, {(1,1)/w_2}, {(2,1)/w_3}, {(3,1)/w_4},  {(5,1)/w_6}, {(6,1)/w_7}, {(7,1)/w_8}, 
          {(3.5,4)/W}}
        \node[vertex] (\name) at \pos {$\name$};
        \foreach \pos/\name in {{(4,1)/w_5}}
            \node[vertex,blue] (\name) at \pos {$\name$}; 

    \tikzstyle{vertex}=[circle,draw,minimum size=9pt,inner sep=1pt]

    \node[ball, at={(0.5,2)}, font=\fontsize{6}{10}\selectfont] (3) {$\{w_1, w_2\}$};
    \node[ball, at={(2.5,2)}, font=\fontsize{6}{10}\selectfont] (4) {$\{w_3,w_4\}$};
    \node[ball, at={(4.5,2)}, font=\fontsize{6}{10}\selectfont, blue] (5) {$\{w_5,w_6\}$};
    \node[ball, at={(6.5,2)}, font=\fontsize{6}{10}\selectfont] (6) {$\{w_7,w_8\}$};
    \node[ball, at={(1.5,3)}, font=\fontsize{6}{10}\selectfont] (1) {$\{w_1,w_2,w_3,w_4\}$};
    \node[ball, at={(5.5,3)}, font=\fontsize{6}{10}\selectfont, blue] (2) {$\{w_5,w_6,w_7,w_8\}$};

    \draw (W)--(1) node[at={(1.5,3.5)}, font=\fontsize{6}{10}\selectfont] {$(0,3)$};
    \draw (W)--(2)[ blue]  node[at={(5.5,3.5)}, font=\fontsize{6}{10}\selectfont] {$(1,3)$};
    \draw (1)--(3)  node[at={(0.4,2.5)}, font=\fontsize{6}{10}\selectfont] {$(00,1)$};
    \draw (1)--(4)  node[at={(2.55,2.5)}, font=\fontsize{6}{10}\selectfont] {$(01,2)$};
    \draw (2)--(5)[ blue]  node[at={(4.4,2.5)}, font=\fontsize{6}{10}\selectfont] {$(10,2)$};
    \draw (2)--(6)  node[at={(6.55,2.5)}, font=\fontsize{6}{10}\selectfont] {$(11,1)$};

    \draw (3)--(w_1)[dashed]  node[at={(0,0.5)}, font=\fontsize{6}{10}\selectfont] {$(000,0)$};
    \draw (3)--(w_2)  node[at={(1,0.5)}, font=\fontsize{6}{10}\selectfont] {$(001,1)$};
    \draw (4)--(w_3)  node[at={(2,0.5)}, font=\fontsize{6}{10}\selectfont] {$(010,1)$};
    \draw (4)--(w_4)  node[at={(3,0.5)}, font=\fontsize{6}{10}\selectfont] {$(011,1)$};
    \draw (5)--(w_5)[ blue]  node[at={(4,0.5)}, font=\fontsize{6}{10}\selectfont] {$(100,1)$};
    \draw (5)--(w_6)  node[at={(5,0.5)}, font=\fontsize{6}{10}\selectfont] {$(101,1)$};
    \draw (6)--(w_7)  node[at={(6,0.5)}, font=\fontsize{6}{10}\selectfont] {$(110,1)$};
    \draw (6)--(w_8)[dashed]  node[at={(7,0.5)}, font=\fontsize{6}{10}\selectfont] {$(111,0)$};
    \draw node[at={(3.5,4.5)}, font=\fontsize{7}{10}\selectfont] {$u_2: \quad (\epsilon,6)$};
    \end{tikzpicture}
    
}
\vspace*{-0.2cm}
    \caption{Neighborhood tree of vertex $u_2 \in U$ that shows the path $\omega$ taken from root to the vertex $w_j \in W$ at the leaf and number of edges $d_{u_2,\omega}$ for each node using the tuple $(\omega, d_{u_2,\omega})$.}
    \label{fig:neighborhoodTree}
\vspace*{-0.2cm}
\end{figure*}

\subsection{\textsf{FM} algorithm}\label{sec:appendix:FM}

In Algorithm~\ref{alg:fm_org} we give a high-level description of the \textsf{FM} algorithm. 
The \textsf{FM} algorithm first constructs \emph{neighborhood trees} (Figure~\ref{fig:neighborhoodTree}) for each vertex~$u_i$ in~$U$, which are
labeled binary trees of depth~$r$, whose nodes are labeled by a bit string $\omega$, where $0 \leq |\omega| \leq r$. 
The root node of a neighborhood tree contains all the vertices in~$W$ and is labeled by the empty string, $\epsilon$. 
The bit string $\omega$ of a node is obtained by starting at the root and following the path in the tree up to the node and concatenating a 0 or a 1 to $\omega$ each time the path visits the left, or the right child, respectively.
The nodes at level~$i$ in the tree correspond to a partition of~$W$ into sets of size~$|W|/2^i$. Thus, each leaf node corresponds to a vertex in~$W$. Each node in the neighborhood tree of $u_i$ stores $d_{i, \omega}$, the number of neighbors of $u_i$ in the set $W_{\omega}$ of vertices of~$W$ that are associated with the node with label $\omega$, i.e., 
$d_{i, \omega} = |N(u_i) \cap W_{\omega}|$, where $N(u_i)$ is the set of the neighbors of $u_i$.
In Figure~\ref{fig:neighborhoodTree}, we show the neighborhood tree of vertex $u_2$ for the bipartite graph with 
$|U| = |W| = n = 8$ given in Figure~\ref{fig:motivation}. Vertex~$w_3$ is a leaf and would have $\omega =  010$ in all neighborhood trees of the graph. \looseness=-1

\begin{algorithm}[!t]
\caption{{\small \textsf{FM} Algorithm: High-level description}}
{
\small
\begin{algorithmic}[1]
\INPUT {$G(U,W,E)$: Bipartite graph;} 
\State{Initialize: clique index, $q \gets 0$; no.\ of vertices $n \gets |U|$; and no.\ of remaining edges $\hat{m} \gets |E|$.}
\State{Build neighborhood trees for vertices $u_i \in U$.} 
\While{$\hat{m} \geq n^{2-\delta}$}
    \State{$\hat{k} \gets \bigg\lfloor{\frac{\delta \log n}{\log(2n^2/\hat{m})} }\bigg\rfloor$}

     \State \parbox[t]{\dimexpr\linewidth-\algorithmicindent}{
   Starting with $t=1$ and $U_t = U$, iterate  while $t \leq \hat{k}$ to select $\hat{k}$ vertices based on $c_0$ and $c_1$ (calculated as in Equation \ref{eq:c01}) to form the ordered set $K_{q} = \{ y_1, y_2, \ldots, y_{t} \}   \subseteq  W$. 
   At each stage $t$, obtain set $U_{t+1}=\left\{u \in U_t \mid y_1, \ldots, y_{t} \in N(u)\right\}$ which omits vertices from $U$ that are not adjacent to $K_{q}$. And update the corresponding neighborhood trees such that $d_{i,\omega}$ of the nodes along the path from root to the selected vertex is decreased by 1 for each neighborhood tree corresponding to $U_t$.
  }
   \State \parbox[t]{\dimexpr\linewidth-\algorithmicindent}{
   Form the $\delta$-clique $C_q = (U_q, K_q)$, where, $U_q  = U_{t}$.
   }
    \State{Update $G$ by removing the edges associated with clique~$C_q$ and update $\hat{m} \gets |E|$.}
    \State{$q \gets q + 1$}
\EndWhile
\State{For each clique in $\mathcal{C} = \{C_1, \ldots, C_q\}$  add edges from all the vertices in the left partition $U_q$ to an additional vertex~$z_q$, and from~$z_q$ to all the vertices in the right partition $K_q$. This would restructure the graph by replacing $|U_q| \times |K_q|$ edges with $|U_q| + |K_q|$ edges.}
\State{The remaining edges in~$G$ are trivial cliques and are added in the restructured graph~$G^*$.}
\State{\textbf{Output:} $G^*$, the restructured graph of~$G$.}
\end{algorithmic}
\label{alg:fm_org}
}
\end{algorithm}

The algorithm computes~$\hat{k}$ and then, in Line~5, it performs~$\hat{k}$ iterations
(indexed by~$t$) to select vertices for the right partition of clique~$C_q$. In each iteration it selects vertices at leaf nodes of the neighborhood trees by following a path based on $d_{i, \omega}$ of child nodes and the number of distinct ordered subsets of the neighborhood tree at each level. In order to select a path at each level, the algorithm calculates $c_{0}$ and $c_{1}$ as follows: 
\begin{equation}
c_{j} = C\left(U_t, \mathscr{C}_{t, \omega \cdot j}\right)=\sum_{i:{u_{i} \in U_{t}}} d_{i, \omega \cdot j} \cdot\left(d_{i}-1\right)^{[{k}-t]},  \, j= 0,1.  \label{eq:c01}
\end{equation}
$c_{j}$ counts the number of ordered sets in $\mathscr{C}_{t, \omega}$ whose elements are adjacent to~$u_{i}$, where $\mathscr{C}_{t, \omega}$ denotes the collection of all ordered sets $K$ of $W$ of size~$k$ in iteration~$t$, with~$t \leq k$. In each iteration, a vertex 
in~$W$ is added to the right partition of the clique. 
If $|N(u_{i})|=d_{i}$, then the number of distinct ordered subsets of $N(u_{i})$ of size~$k$ is~$d_{i}^{[k]}$, where  $d_{i}^{[k]} = d_{i}(d_{i}-1)(d_{i}-2) \cdots(d_{i}-k+1)$. Based on $c_{j}$'s, the \textsf{FM} algorithm chooses either the left node (if $c_{0} \geq c_{1}$) 
or the right node (if $c_{0} < c_{1}$). When the algorithm reaches a leaf node, it selects the corresponding vertex in~$W$, updates the set of common neighbors ($U_{t}$), and decreases $d_{i, \omega}$ of each node along the path to the selected vertex by~1, in  neighborhood trees corresponding to $U_{t}$.
Thus, it guarantees the selection of a new vertex in the following iterations. The algorithm continues this process until it extracts~$\hat{k}$ unique vertices to form the set~$K_q$, the right partition of the $q$-th $\delta$-clique~$C_q$.

After forming set $K_q$, the \textsf{FM} algorithm proceeds to identify the common neighbors of the vertices that are part of set~$K_q$ in order to construct~$U_q$, which is the left partition of the $q$-th $\delta$-clique~$C_q$. 
The pair~$(K_q, U_q)$ forms a $\delta$-clique~$C_q$. Subsequently, the algorithm updates~$\hat{m}$, the number of remaining edges in the graph after extracting each~$\delta$-clique. This update is done by subtracting the product of the sizes of $K_q$ and $U_q$ from $\hat{m}$, i.e., $|U_q| \times |K_q|$. The algorithm then updates~$\hat{k}$ considering the updated value of~$\hat{m}$ and repeats the procedure to extract more $\delta$-cliques until no more $\delta$-cliques can be found (determined by the condition of the while loop). The edges that are not removed in this process form trivial cliques.

The \textsf{FM} algorithm restructures the graph by adding a new vertex set~$Z$ to the bipartite graph, thus converting it into a tripartite graph (Figure~\ref{fig:motivation}a), where each vertex in~$Z$ corresponds to one clique extracted by the algorithm. Each vertex of the left and right partitions of $\delta$-clique $C_q$ is then connected via an edge to a new vertex $z_q \in Z$, thus forming a tripartite graph.  This decreases the number of edges from $|U_q| \times |K_q|$ to $|U_q| + |K_q|$, thus reducing the number of edges in the graph. The restructured graph obtained by the algorithm preserves the path information of the original graph. The running time of the \textsf{FM} algorithm is~$O(mn^{\delta} \log^2 n)$. The \textsf{FM} algorithm can be \emph{extended to the case of non-bipartite graphs}, where it restructures a graph in time $O(mn^{\delta} \log^2 n)$, as shown by Feder and Motwani~\cite{federMotwani}.

\subsection{Example: \textsf{CPGR} Execution}\label{sec:appendix:example}
We now show how~\textsf{CPGR} works on a bipartite graph~$G$ with partitions~$U$ and~$W$, $n = |U| = |W| = 8$, and $|E| = 54$, shown in Figure~\ref{fig:given_graph(a)}. 
\textsf{CPGR} (Algorithm \ref{alg:CPGR}), first initializes $q = 0, n = |W| = |U| = 8 , \hat{m} = |E| = 54, \mathcal{C} = \emptyset$, and $d_w = [0]_n $. It then calculates the degree of each vertex $d_{w_j}$ in Lines 6-8, as shown in Table~\ref{table:1}. It calculates~$\hat{k}(n,m,\delta) = 2$ in Line~9. Thus, the condition for the while loop in Line~10 is met and \textsf{CPGR} calls \textsf{CSA} in Line~11. With the given inputs, \textsf{CSA} sorts~$d_w$ in non-increasing order. Let $\{w_4, w_2, w_3, w_5, w_6, w_1, w_7, w_8\}$ be the non-increasing order with corresponding degrees given as $d_{w_{\pi(j)}}$ in Table~\ref{table:1}. \looseness=-1 

\begin{table}[h!]
\centering
{\small
 \begin{tabular}{c c c c c c c c c c} 
 \hline
 \textsf{CPGR} step & $j$ & 1 & 2 & 3 & 4 & 5 & 6 & 7 & 8\\ 
  \hline
 Initialization & $d_{w_j}$ & 6 &  7 & 7 & 8 & 7 & 7 & 6 & 6\\ 
 Sorted $d_{w_j}$ & $d_{w_{\pi(j)}} $ & 8 &  7 & 7 & 7 & 7 & 6 & 6 & 6\\ 
After $1^{st}$iteration &  $d_{w_j}$ & 6 &  0 & 0 & 1 & 0 & 7 & 6 & 6\\ 
 \hline
\end{tabular}
\vspace*{0.15cm}
\caption{Degrees of vertices in partition $W$ at different steps of \textsf{CPGR}.}
\label{table:1}
}
\end{table}

\tikzstyle{vertex}=[circle,draw,minimum size=12pt,inner sep=0pt]
\tikzstyle{edge} = [draw,thick,-]
\usetikzlibrary{decorations.markings}
\usetikzlibrary{shapes,decorations,arrows,calc,arrows.meta,fit,positioning}
\tikzset{edge/.style = {->,> = latex'}}

\begin{figure*}[!t]
\centerline{
\subfloat[\centering]{

        \begin{tikzpicture}[scale=0.6, auto, swap]
	\tikzset{edge/.style = {->,> = latex'}}
    		\foreach \pos/\name in {{(0,7)/u_1}, {(0,6)/u_2}, {(0,5)/u_3}, {(0,4)/u_4}, {(0,3)/u_5}, {(0,2)/u_6}, {(0,1)/u_7}, {(0,0)/u_8},
      {(2.6,7)/w_1}, {(2.6,6)/w_2}, {(2.6,5)/w_3}, {(2.6,4)/w_4}, {(2.6,3)/w_5}, {(2.6,2)/w_6}, {(2.6,1)/w_7}, {(2.6,0)/w_8}}
        	\node[vertex] (\name) at \pos {$\name$};
        \draw (u_1)--(w_1);
        \draw (u_1)--(w_2);
        \draw (u_1)--(w_3);
        \draw (u_1)--(w_4);
        \draw (u_1)--(w_5);
        \draw (u_1)--(w_6);

        \draw (u_2)--(w_2);
        \draw (u_2)--(w_3);
        \draw (u_2)--(w_4);
        \draw (u_2)--(w_5);
        \draw (u_2)--(w_6);
        \draw (u_2)--(w_7);

        \draw (u_3)--(w_2);
        \draw (u_3)--(w_3);
        \draw (u_3)--(w_4);
        \draw (u_3)--(w_5);
        \draw (u_3)--(w_6);
        \draw (u_3)--(w_7);
        \draw (u_3)--(w_8);

        \draw (u_4)--(w_1);
        \draw (u_4)--(w_2);
        \draw (u_4)--(w_3);
        \draw (u_4)--(w_4);
        \draw (u_4)--(w_5);
        \draw (u_4)--(w_6);
        \draw (u_4)--(w_7);
        \draw (u_4)--(w_8);

        \draw (u_5)--(w_8);
        \draw (u_5)--(w_1);
        \draw (u_5)--(w_2);
        \draw (u_5)--(w_3);
        \draw (u_5)--(w_4);
        \draw (u_5)--(w_5);
        \draw (u_5)--(w_6);

        \draw (u_6)--(w_4);
        \draw (u_6)--(w_7);
        \draw (u_6)--(w_8);
        \draw (u_6)--(w_1);
        \draw (u_6)--(w_2);

        \draw (u_7)--(w_3);
        \draw (u_7)--(w_4);
        \draw (u_7)--(w_5);
        \draw (u_7)--(w_6);
        \draw (u_7)--(w_7);
        \draw (u_7)--(w_8);
        \draw (u_7)--(w_1);

        \draw (u_8)--(w_1);
        \draw (u_8)--(w_2);
        \draw (u_8)--(w_3);
        \draw (u_8)--(w_4);
        \draw (u_8)--(w_5);
        \draw (u_8)--(w_6);
        \draw (u_8)--(w_7);
        \draw (u_8)--(w_8);            		
 	\end{tikzpicture}
 	\label{fig:given_graph(a)}
}
\subfloat[\centering]{
	\begin{tikzpicture}[scale=0.6, auto, swap]
	\tikzset{edge/.style = {->,> = latex'}}
    		\foreach \pos/\name in {{(0,7)/u_1}, {(0,6)/u_2}, {(0,5)/u_3}, {(0,4)/u_4}, {(0,3)/u_5}, {(0,2)/u_6}, {(0,1)/u_7}, {(0,0)/u_8},
      {(2.6,7)/w_1}, {(2.6,6)/w_2}, {(2.6,5)/w_3}, {(2.6,4)/w_4}, {(2.6,3)/w_5}, {(2.6,2)/w_6}, {(2.6,1)/w_7}, {(2.6,0)/w_8}}
        	\node[vertex] (\name) at \pos {$\name$};

        \draw (u_1)--(w_1);
        \draw [blue,line width=0.5mm] (u_1)--(w_2);
        \draw (u_1)--(w_3);
        \draw [blue,line width=0.5mm] (u_1)--(w_4);
        \draw (u_1)--(w_5);
        \draw (u_1)--(w_6);

        \draw [blue,line width=0.5mm] (u_2)--(w_2);
        \draw (u_2)--(w_3);
        \draw [blue,line width=0.5mm] (u_2)--(w_4);
        \draw (u_2)--(w_5);
        \draw (u_2)--(w_6);
        \draw (u_2)--(w_7);

        \draw [blue,line width=0.5mm](u_3)--(w_2);
        \draw (u_3)--(w_3);
        \draw [blue,line width=0.5mm](u_3)--(w_4);
        \draw (u_3)--(w_5);
        \draw (u_3)--(w_6);
        \draw (u_3)--(w_7);
        \draw (u_3)--(w_8);

        \draw (u_4)--(w_1);
        \draw [blue,line width=0.5mm](u_4)--(w_2);
        \draw (u_4)--(w_3);
        \draw [blue,line width=0.5mm](u_4)--(w_4);
        \draw (u_4)--(w_5);
        \draw (u_4)--(w_6);
        \draw (u_4)--(w_7);
        \draw (u_4)--(w_8);

        \draw (u_5)--(w_8);
        \draw (u_5)--(w_1);
        \draw [blue,line width=0.5mm](u_5)--(w_2);
        \draw (u_5)--(w_3);
        \draw [blue,line width=0.5mm](u_5)--(w_4);
        \draw (u_5)--(w_5);
        \draw (u_5)--(w_6);

        \draw [blue,line width=0.5mm](u_6)--(w_4);
        \draw (u_6)--(w_7);
        \draw (u_6)--(w_8);
        \draw (u_6)--(w_1);
        \draw [blue,line width=0.5mm](u_6)--(w_2);

        \draw (u_7)--(w_3);
        \draw (u_7)--(w_4);
        \draw (u_7)--(w_5);
        \draw (u_7)--(w_6);
        \draw (u_7)--(w_7);
        \draw (u_7)--(w_8);
        \draw (u_7)--(w_1);

        \draw (u_8)--(w_1);
        \draw [blue,line width=0.5mm](u_8)--(w_2);
        \draw (u_8)--(w_3);
        \draw [blue,line width=0.5mm](u_8)--(w_4);
        \draw (u_8)--(w_5);
        \draw (u_8)--(w_6);
        \draw (u_8)--(w_7);
        \draw (u_8)--(w_8);
  		
 	\end{tikzpicture}
 	\label{fig:ch-matching-bip:bm2(b)}
}
\subfloat[\centering]{
	\begin{tikzpicture}[scale=0.6, auto, swap]
	\tikzset{edge/.style = {->,> = latex'}}
    		\foreach \pos/\name in {{(0,7)/u_1}, {(0,6)/u_2}, {(0,5)/u_3}, {(0,4)/u_4}, {(0,3)/u_5}, {(0,2)/u_6}, {(0,1)/u_7}, {(0,0)/u_8},
      {(2.6,7)/w_1}, {(2.6,6)/w_2}, {(2.6,5)/w_3}, {(2.6,4)/w_4}, {(2.6,3)/w_5}, {(2.6,2)/w_6}, {(2.6,1)/w_7}, {(2.6,0)/w_8}}
        	\node[vertex] (\name) at \pos {$\name$};
        \draw (u_1)--(w_1);
        \draw [red,line width=0.5mm](u_1)--(w_3);
        \draw [red,line width=0.5mm](u_1)--(w_5);
        \draw (u_1)--(w_6);

        \draw [red,line width=0.5mm](u_2)--(w_3);
        \draw [red,line width=0.5mm](u_2)--(w_5);
        \draw (u_2)--(w_6);
        \draw (u_2)--(w_7);

        \draw [red,line width=0.5mm](u_3)--(w_3);
        \draw [red,line width=0.5mm](u_3)--(w_5);
        \draw (u_3)--(w_6);
        \draw (u_3)--(w_7);
        \draw (u_3)--(w_8);

        \draw (u_4)--(w_1);
        \draw [red,line width=0.5mm](u_4)--(w_3);
        \draw [red,line width=0.5mm](u_4)--(w_5);
        \draw (u_4)--(w_6);
        \draw (u_4)--(w_7);
        \draw (u_4)--(w_8);

        \draw (u_5)--(w_8);
        \draw (u_5)--(w_1);
        \draw [red,line width=0.5mm](u_5)--(w_3);
        \draw [red,line width=0.5mm](u_5)--(w_5);
        \draw (u_5)--(w_6);

        \draw (u_6)--(w_7);
        \draw (u_6)--(w_8);
        \draw (u_6)--(w_1);

        \draw [red,line width=0.5mm](u_7)--(w_3);
        \draw (u_7)--(w_4);
        \draw [red,line width=0.5mm](u_7)--(w_5);
        \draw (u_7)--(w_6);
        \draw (u_7)--(w_7);
        \draw (u_7)--(w_8);
        \draw (u_7)--(w_1);

        \draw (u_8)--(w_1);
        \draw [red,line width=0.5mm](u_8)--(w_3);
        \draw [red,line width=0.5mm](u_8)--(w_5);
        \draw (u_8)--(w_6);
        \draw (u_8)--(w_7);
        \draw (u_8)--(w_8);

 	\end{tikzpicture}
 	\label{fig:ch-matching-bip:bm2(c)}
}

\subfloat[\centering]{
\centering
	\begin{tikzpicture}[scale=0.6, auto, swap]
	\tikzset{edge/.style = {->,> = latex'}}
    		\foreach \pos/\name in {{(0,7)/u_1}, {(0,6)/u_2}, {(0,5)/u_3}, {(0,4)/u_4}, {(0,3)/u_5}, {(0,2)/u_6}, {(0,1)/u_7}, {(0,0)/u_8},
      {(2.6,7)/w_1}, {(2.6,6)/w_2}, {(2.6,5)/w_3}, {(2.6,4)/w_4}, {(2.6,3)/w_5}, {(2.6,2)/w_6}, {(2.6,1)/w_7}, {(2.6,0)/w_8}}
        	\node[vertex] (\name) at \pos {$\name$};
        	


        \draw (u_1)--(w_1);
        \draw (u_1)--(w_6);

        \draw (u_2)--(w_6);
        \draw (u_2)--(w_7);

        \draw (u_3)--(w_7);
        \draw (u_3)--(w_6);
        \draw (u_3)--(w_8);

        \draw (u_4)--(w_1);
        \draw (u_4)--(w_7);
        \draw (u_4)--(w_6);
        \draw (u_4)--(w_8);

        \draw (u_5)--(w_8);
        \draw (u_5)--(w_1);
        \draw (u_5)--(w_6);

        \draw (u_6)--(w_7);
        \draw (u_6)--(w_8);
        \draw (u_6)--(w_1);

        \draw (u_7)--(w_4);
        \draw (u_7)--(w_7);
        \draw (u_7)--(w_8);
        \draw (u_7)--(w_1);
        \draw (u_7)--(w_6);

        \draw (u_8)--(w_1);
        \draw (u_8)--(w_7);
        \draw (u_8)--(w_6);
        \draw (u_8)--(w_8);

 	\end{tikzpicture}
 	\label{fig:ch-matching-bip:bm2(d)}
}


 \subfloat[\centering]{
	\begin{tikzpicture}[scale=0.6, auto, swap]
	\tikzset{edge/.style = {->,> = latex'}}
    		\foreach \pos/\name in {{(0,7)/u_1}, {(0,6)/u_2}, {(0,5)/u_3}, {(0,4)/u_4}, {(0,3)/u_5}, {(0,2)/u_6}, {(0,1)/u_7}, {(0,0)/u_8}, 
      {(3,7)/w_1}, {(3,6)/w_2}, {(3,5)/w_3}, {(3,4)/w_4}, {(3,3)/w_5}, {(3,2)/w_6}, {(3,1)/w_7}, {(3,0)/w_8}}
        	\node[vertex] (\name) at \pos {$\name$};


        \draw (u_1)--(w_1);

        \draw   (u_2)--(w_7);

        \draw  (u_3)--(w_7);
        \draw  (u_3)--(w_8);

        \draw   (u_4)--(w_1);
        \draw  (u_4)--(w_7);
        \draw  (u_4)--(w_8);

        \draw  (u_5)--(w_8);
        \draw  (u_5)--(w_1);

        \draw  (u_6)--(w_7);
        \draw  (u_6)--(w_8);
        \draw  (u_6)--(w_1);

        \draw  (u_7)--(w_4);
        \draw  (u_7)--(w_7);
        \draw  (u_7)--(w_8);
        \draw  (u_7)--(w_1);

        \draw  (u_8)--(w_1);
        \draw  (u_8)--(w_7);
        \draw  (u_8)--(w_8);


        \draw (u_2)--(w_6);

        \draw (u_3)--(w_6);

        \draw (u_4)--(w_6);

        \draw (u_5)--(w_6);

        \draw (u_7)--(w_6);

        \draw (u_8)--(w_6);
        

        \draw (u_1)--(w_6);


            \foreach \pos/\name in {{(1.5,4.5)/z_1}, {(1.5,3)/z_2}}
                \node[vertex,fill=gray] (\name) at \pos {$\name$};
        \draw [red,line width=0.5mm](u_1)--(z_2);
        \draw [red,line width=0.5mm](u_2)--(z_2);
        \draw [red,line width=0.5mm](u_3)--(z_2);

        \draw [red,line width=0.5mm](u_4)--(z_2);

        \draw [red,line width=0.5mm](u_5)--(z_2);

        \draw [red,line width=0.5mm](u_7)--(z_2);

        \draw [red,line width=0.5mm](u_8)--(z_2);
        
        \draw [red,line width=0.5mm](z_2)--(w_3);
        \draw [red,line width=0.5mm](z_2)--(w_5);

                \draw [blue,line width=0.5mm] (u_1)--(z_1);

        \draw [blue,line width=0.5mm] (u_2)--(z_1);

        \draw [blue,line width=0.5mm](u_3)--(z_1);

        \draw [blue,line width=0.5mm](u_4)--(z_1);

        \draw [blue,line width=0.5mm](u_5)--(z_1);

        \draw [blue,line width=0.5mm](u_6)--(z_1);

        \draw [blue,line width=0.5mm](u_8)--(z_1);
        
        \draw [blue,line width=0.5mm](z_1)--(w_2);
        \draw [blue,line width=0.5mm](z_1)--(w_4);

 	\end{tikzpicture}
 	\label{fig:ch-matching-bip:bm2(e)}}

}
\vspace*{-0.2cm}
\caption{\small Clique partitioning: (a) given bipartite graph $G(U,W)$; (b) extracted clique $C_1$; (c) extracted clique $C_2$; (d) graph with trivial cliques; and (e) restructured tripartite graph $G^*(U,Z,W)$}
\label{fig:ch-mathing-bip:bm(caption)}
\end{figure*}

In Line 5, \textsf{CSA} forms set $\mathcal{K}$, where $\mathcal{K} = \{ w_4, w_2, w_3, w_5, w_6 \}$ has all 
vertices whose degrees are greater than or equal to $d_{w_{\pi(\hat{k})}}=7$. This results in $\gamma = 2$ 
in Line~7, therefore \textsf{CSA} extracts two cliques in Lines 8-13.
Since~$q$ is initialized to~0 in \textsf{CPGR}, the for loop in Line~8 iterates two times for~$c=1$ and~$2$. In Line~9, \textsf{CSA} forms $K_1 = \{w_4, w_2\}$, and in Line 10, it forms the left partition $U_{K_{1}}$ of bipartite clique~$\mathcal{C}_1$, by selecting the common neighbors of the vertices in $K_{1}$. For example, for set $K_1 = \{w_4, w_2\}$, $U_{K_1} = U - \{u_7\}$ as $w_4$ and $w_2$ have $\{u_1, u_2, u_3, u_4, u_5, u_6, u_8\}$ as common neighbors. In Line~12, \textsf{CSA} updates the adjacency matrix~$A$ by removing the edges of bipartite clique~$\mathcal{C}_1$ from the original bipartite graph~$G$ and finally in Line~13, it updates all~$d_w$ by subtracting $|U_{K_1}| = 7$ from both $d_{w_{4}}$ and $d_{w_{2}}$, thus $d_{w_{4}}=1$ and $d_{w_{2}}=0$.
Similarly, when $c=2$, $K_2 = \{w_3, w_5\}$ and $U_{K_2} = U - \{u_6\}$. \textsf{CSA} updates the adjacency matrix $A$ and $d_w$ such that $d_{w_{3}}=d_{w_{5}}=0$. 
Therefore, in the first execution, \textsf{CSA}, forms two bipartite cliques and removes a total of $(7 \times 2) + (7 \times 2)  = 28$ edges from~$G$. The updated degrees of the vertices are shown in Table~\ref{table:1}. 
It then forms the set of bipartite cliques extracted in the current execution in Line~14, $\mathscr{C} =\{C_{1}, C_{2}\}$. The two cliques are shown in Figures~\ref{fig:ch-matching-bip:bm2(b)} and~\ref{fig:ch-matching-bip:bm2(c)}.
\looseness=-1

\textsf{CSA} returns $(\mathscr{C}, \hat{A}, \hat{d}_w)$ and then \textsf{CPGR} updates the set $\mathscr{C} =\{C_{1}, C_{2}\}$, $d_w$, the adjacency matrix~$A$, and the number of edges $\hat{m}= 54-28 =26$, in Lines~\mbox{12-15}. In Line 16, \textsf{CPGR} updates $q=2$ and finally, in Line 17 it updates~$\hat{k}$ according to the new value of~$\hat{m}$ in Line 15 which results in $\hat{k} =1$. Since $\hat{k} =1$, it means \textsf{CPGR} extracts only trivial bipartite cliques (shown in Figure~\ref{fig:ch-matching-bip:bm2(d)}) which do not contribute to reduction in the number of edges in the restructured graph. Thus, it does not meet the condition in the while loop (Line 6) and \textsf{CPGR} terminates.

\textsf{CPGR} then restructures the graph by adding two vertices, $z_1$, and  $z_2$, corresponding to the two cliques $C_1$ and $C_2$, and adds the corresponding edges to form the tripartite graph. The edges in the given graph $G(U,W,E)$, that are not part of any $\delta$-clique are connected directly as shown in Figure~\ref{fig:ch-matching-bip:bm2(e)}.

\end{document}